\documentclass[a4paper]{article}

\usepackage[utf8]{inputenc}
\usepackage[american]{babel}
\usepackage{amsmath,amstext,amsfonts,amssymb,amsthm,xfrac}
\usepackage{paralist}
\usepackage{url,xspace}
\usepackage[ruled,vlined,fillcomment,noend,linesnumbered]{algorithm2e}
\usepackage{ifthen}
\usepackage{subfig}
\usepackage{graphicx}
\graphicspath{{fig/}}
\usepackage{fullpage}
\usepackage[shadow]{todonotes}

\usepackage{figlatexjpdc}

\usepackage{tikz-timing}
\usetikztiminglibrary[new={char=Q,reset char=R}]{counters}
\usetikzlibrary{patterns}

\newcommand{\period}{\ensuremath{T}\xspace}
\newcommand{\T}{\ensuremath{T}\xspace}
\newcommand{\Nfaults}{\ensuremath{N_{\text{faults}}}\xspace}

\newcommand{\ttopt}{\ensuremath{T_{\text{opt}}}\xspace}
\newcommand{\Cr}{\ensuremath{C}\xspace}
\newcommand{\Cp}{\ensuremath{C_{p}}\xspace}
\newcommand{\D}{\ensuremath{D}\xspace}
\newcommand{\R}{\ensuremath{R}\xspace}
\newcommand{\p}{\ensuremath{p}\xspace}

\newcommand{\Nckpt}{\ensuremath{N_{\text{ckpt}}}\xspace}
\newcommand{\recall}{\ensuremath{r}\xspace}
\newcommand{\precision}{\ensuremath{p}\xspace}
\newcommand{\trust}{\ensuremath{q}\xspace}
\newcommand{\muP}{\ensuremath{\mu_{\text{P}}}\xspace}
\newcommand{\muNP}{\ensuremath{\mu_{\text{NP}}}\xspace}
\newcommand{\munew}{\ensuremath{\mu_{\text{e}}}\xspace}
\newcommand{\waste}{\ensuremath{\textsc{Waste}}\xspace}

\newcommand{\esperance}[1]{\mathbb E\left(#1\right)}

\newcommand{\Time}[1][]{\ensuremath{\textsc{Time}_{\text{#1}}}\xspace}
\newcommand{\Waste}[1][]{\ensuremath{\textsc{Waste}_{\text{#1}}}\xspace}

\newcommand{\extr}{\ensuremath{\text{extr}}}

\newcommand{\Tp}{\ensuremath{T_{\textsc{Pred}}}\xspace}
\newcommand{\Ty}{\ensuremath{T_{\textsc{NoPred}}}\xspace}
\newcommand{\Tyafo}{\ensuremath{T_{\newdaly}}\xspace}
\newcommand{\Tlost}[1][]{\ensuremath{T_{\text{lost}#1}}\xspace}

\newcommand{\Te}{\ensuremath{T_{\extr}}\xspace}

\newcommand{\Wregular}{\ensuremath{W_{\mathit{reg}}}\xspace}

\newcommand{\Instant}{\textsc{Instant}\xspace}

\newcommand{\young}{\textsc{Young}\xspace}
\newcommand{\daly}{\textsc{Daly}\xspace}
\newcommand{\newdaly}{\textsc{RFO}\xspace}

\newcommand{\OptimalPrediction}{\textsc{OptimalPrediction}\xspace}
\newcommand{\inexact}{\textsc{InexactPrediction}\xspace}
\newcommand{\bestper}{\textsc{BestPeriod}\xspace}

\newcommand{\periodicignoringI}{\Instant}

\newcommand{\bestnewdaly}{\bestper \newdaly\xspace}

\newcommand{\bestperiodicignoringI}{\bestper\Instant \xspace}

\newcommand{\bestperiodicOptimalprediction}{\bestper \OptimalPrediction\xspace}

\newtheorem{theorem}{Theorem}
\newtheorem{proposition}{Proposition}

\theoremstyle{definition}

\newcommand{\faultbis}[1]{
\draw[<-, color=red] (#1) -- ($(#1)+(0.2,1.2)$) -- ($(#1)+(0.1,1.4)$) --  ($(#1)+(0.2,2)$) node[above, left] {\scriptsize{fault}};
} 
\newcommand{\predfault}[1]{
\draw[<-, color=blue] ($(#1) + (0,1)$) -- ($(#1)+(0.2,1.6)$)  -- ($(#1)+(0.1,1.7)$) -- ($(#1)+(0.2,2)$)  node[above, right] {\scriptsize{Predicted fault}};
} 

\newcommand{\legende}[3]{
\draw[thick, <->] ($(#1)+(0,-0.20)$) -- ($(#1)+(#2,-0.20)$) node[below=-0.5pt, midway] {\scriptsize{#3}};
}

\newcommand{\arrowtime}[2]{
\draw[thick, color=black,->] (0,#1) -- (#2,#1) node[below=-0.5pt, ] {\scriptsize{Time}};
}

\newcommand{\ttrd}{4} 

\begin{document}

\title{Checkpointing algorithms and fault prediction}

\author{Guillaume Aupy$^{1,4}$,Yves Robert$^{1,3,4}$, Fr\'ed\'eric Vivien$^{2,4}$ and Dounia Zaidouni$^{2,4}$\\
 $1.$ \'Ecole Normale Sup\'erieure de Lyon\\
 $2.$ INRIA, France \\
 $3.$ University of Tennessee Knoxville, USA\\
 $4.$ LIP - Université de Lyon - CNRS : UMR5668 - INRIA - \\École Normale Sup\'erieure de Lyon- Université Claude Bernard - Lyon, France\\
 \url{{Guillaume.Aupy | Yves.Robert | Frederic.Vivien | Dounia.Zaidouni}@ens-lyon.fr}
 }

\maketitle
\begin{abstract}
This paper deals with the impact of fault prediction techniques on checkpointing strategies.
We extend the classical first-order analysis of Young and Daly in the presence of a fault prediction system,
characterized by its recall and its precision. In this framework, 
we provide optimal algorithms to decide whether and when to
take predictions into account, and we derive the optimal value of the checkpointing period.
These results allow to analytically assess the
key parameters that impact the performance of  fault predictors at very large scale.
\end{abstract}


\section{Introduction}
\label{sec.intro}

Nowadays, the most powerful High Performance Computing systems
experience about one fault per day~\cite{6264677,TsubameSC12}. Consider
the relative slopes describing the evolution of the reliability of
individual components on one side, and the evolution of the number of
components on the other side:  the reliability of an entire platform is
expected to decrease, due to probabilistic amplification, as its
number of components increases.  Therefore, applications running on
large computing systems have to cope with platform faults. There are
two main approaches. On the one hand, applications can use fault-tolerance
mechanisms such as checkpoint and rollback in order to become
resilient. On the other hand, system administrators can try to predict
where and when faults will strike. Although considerable research has
been devoted to fault
predictors~\cite{Fulp:2008:PCS:1855886.1855891,GainaruIPDPS12,GainaruSC12,LiangZXS07,5958823,5542627},
no predictor will ever be able to predict every fault. Therefore, fault
predictors will have to be used in conjunction with fault-tolerance
mechanisms.

In this paper, we assess the impact of fault prediction techniques on checkpointing strategies.
We assume to have jobs executing on a platform subject to faults,
and we let $\mu$ be the Mean Time Between Faults (MTBF) of the platform.
In the absence of fault prediction, the standard approach is to take periodic checkpoints, each of
length \Cr, every period of duration \period. In steady-state utilization of the platform,
the value \ttopt of \period that minimizes the expected waste of resource usage due to checkpointing
is approximated
as  $\ttopt = \sqrt{2 \mu\Cr }+\Cr$, or $\ttopt = \sqrt{2 (\mu +\R)\Cr }+\Cr$ (where \R is the duration of the recovery).
The former expression is the well-known Young's formula~\cite{young74},
while the latter is due to Daly~\cite{daly04}.

Now, when some fault prediction mechanism is available, can we compute a better checkpointing period
to decrease the expected waste? and to what extent? Critical parameters that characterize a fault prediction
system are its recall \recall, which is  the fraction of faults that are indeed predicted, and its precision \precision,
which is the fraction of predictions that are correct (i.e., correspond to actual faults). 
The major objective of this paper  is to refine the expression of the
expected waste as a function of these new parameters, and to 
design efficient checkpointing policies that take predictions into account.
The key contributions of this paper are:
\begin{compactitem}
\item A refined first-order
  analysis in the absence of fault prediction. It leads to similar performance to 
  Young~\cite{young74} and Daly~\cite{daly04} when faults follow an
  Exponential distribution, and to better performance when faults follow
  a Weibull distribution.
  
  \item The extension of this analysis to fault predictions, and the design of new checkpointing policies that takes
  optimal decisions on whether and when to take these predictions
  into account (or to ignore them).
  
  \item For policies where the decision to trust the predictor is taken with
  the same probability throughout the checkpointing period, we show that we should always trust the predictor, or never, depending upon platform and predictor parameters.
  
  \item For policies where the decision to trust the predictor is taken with
variable probability during the checkpointing period, we show that we should
change strategy only once in the period, moving from never trusting the predictor
when the prediction arrives in the beginning of the period, to always
trusting the predictor
when the prediction arrives later on in the period, and we determine the optimal break-even point.

\item For all policies, we compute the optimal value of the checkpointing period
thereby designing optimal algorithms to minimize the waste when coupling
checkpointing with predictions.

\item An extensive set of simulations that corroborates all
  mathematical derivations.  These simulations are based on
  synthetic fault traces (for Exponential fault distributions,
  and for more realistic Weibull fault distributions) and on log-based
  fault traces. In addition, they include exact prediction dates and
  uncertainty intervals for these dates.
\end{compactitem}

The rest of the paper is organized as follows. We first detail the framework in Section~\ref{sec.framework}. We revisit Young and
Daly's approach in Section~\ref{sec.youngdaly}. We provide optimal algorithms 
to account for  predictions in Section~\ref{sec.no.intervals}: we start with 
simpler policies where the decision to trust the predictor is taken with
  the same probability throughout the checkpointing period (Section~\ref{sec.algo.simple}) before dealing with the most general approach
where the decision to trust the predictor is taken with
variable probability during the checkpointing period (Section~\ref{sec.algo.beautiful}).
Section~\ref{sec.simulations} is 
devoted to simulations: we first describe the framework (Section~\ref{sec.simulations.framework}) and then discuss synthetic and log-based
failure traces in Sections~\ref{sec.simulations.synthetic} and~\ref{sec:logbased} 
respectively. We discuss related work in Section~\ref{sec.related}.
Finally, we provide concluding remarks in Section~\ref{sec.conclusion}.

\begin{table}[bh]
  \centering
  \begin{tabular}{rp{15cm}}
\precision & Predictor precision: proportion of true positives among the number of pre\-dic\-ted faults\\
\recall & Predictor recall: proportion of predicted faults among total number of faults \\
\trust & Probability to trust the predictor\\
MTBF & Mean Time Between Faults \\
$N$ & Number of processors in the platform \\
$\mu$ & Platform MTBF \\
$\mu_{\text{ind}}$ & Individual MTBF \\
\muP & Rate of predicted faults \\
\muNP & Rate of unpredicted faults \\
\munew & Rate of events (predictions or unpredicted faults) \\
\D & Downtime \\
\R & Recovery time\\
\Cr & Duration of a regular checkpoint \\
\Cp & Duration of a proactive checkpoint \\
\period & Duration of a period \\
  \end{tabular}
  \caption{Table of main notations.\label{tab.notations}}
\end{table}

\section{Framework}
\label{sec.framework}

\subsection{Checkpointing strategy}
 
We consider a \emph{platform} subject to faults. Our work is agnostic
of the granularity of the platform, which may consist either of a
single processor, or of several processors that work concurrently and
use coordinated checkpointing.   \emph{Checkpoints} are taken at regular intervals, or
periods, of length \period. We denote by \Cr the duration of a
checkpoint (all checkpoints have same duration).  
By construction, we must enforce that $\Cr \leq \period$.
When a fault strikes the platform, the application is lacking some
resource for a certain period of time of length $\D$, the
\emph{downtime}. The downtime accounts for software rejuvenation
(i.e., rebooting~\cite{875631,1663301}) or for the replacement of the
failed hardware component by a spare one. Then, the application
recovers from the last checkpoint. \R denotes the duration of this
\emph{recovery} time.

\subsection{Fault predictor}

A fault predictor is a mechanism that is able to predict that some faults will take place, 
either at a certain point in time, or within some time-interval window. In this paper,
we assume that the predictor is able to provide exact prediction dates, 
and to generate such predictions early enough so that a \emph{proactive} checkpoint can indeed be taken before 
the event.

The accuracy of the fault predictor is characterized by two quantities, the \emph{recall}
and the \emph{precision}. The recall \recall is the fraction of faults that are predicted while
the precision \precision is the fraction of fault predictions that are correct.
Traditionally, one defines three types of \emph{events}: (i) \textit{True positive} events are faults that the predictor has been able to predict (let $\textit{True}_P$ be their number); (ii)
\textit{False positive} events are fault predictions that did not materialize as actual faults (let $\textit{False}_P$ be their number);
and (iii)  \textit{False negative} events are faults that were not predicted (let $\textit{False}_N$ be their number).
With these definitions, we have
$\recall = \frac{\textit{True}_P}{\textit{True}_P+\textit{False}_N}$ 
and $\p = \frac{\textit{True}_P}{\textit{True}_P+ \textit{False}_P}$.

Proactive checkpoints may have a
different length $\Cp$ than regular checkpoints of length $\Cr$. In
fact there are many scenarios. On the one hand, we may well have $\Cp > \Cr$ in scenarios where regular checkpoints are taken at time-steps where the application memory footprint is minimal~\cite{Hong01}; on the contrary, proactive checkpoints are taken according to
predictions that can take place at arbitrary instants. On the other hand,
we may  have $\Cp < \Cr$ in other scenarios~\cite{5542627}, e.g., when the prediction is
localized to a particular resource subset, hence allowing for a
smaller volume of checkpointed data. 

To keep full generality, we deal with two checkpoint sizes in this paper: \Cr for periodic
checkpoints, and \Cp for proactive checkpoints (those taken upon predictions). 

In the literature, the \emph{lead time} is the interval between the date at which the prediction is made available,
and the actual prediction date. While the lead time is an important parameter,
the shape of its distribution law is irrelevant to the problem: either
a fault is predicted at least $\Cp$ seconds in advance, and then one can checkpoint just in time before the fault, or the prediction is useless!
In other words, predictions that come too late should be classified as unpredicted faults whenever they materialize as actual faults, leading to a smaller value of the predictor recall.

\subsection{Fault rates}

The key parameter is $\mu$, the MTBF of the platform.  If the platform is made
of $N$ components whose individual MTBF is $\mu_{\text{ind}}$, then $\mu =
\frac{\mu_{\text{ind}}}{N}$. This result is true regardless of the fault distribution law\footnote{For the sake of completeness,
we provide a proof of this widely-used result in~\ref{app.mu}. To the best of our knowledge, no proof 
has been published in the literature yet.}.

In addition to $\mu$, the platform MTBF, 
let $\muP$ be the mean time between predicted events (both true positive and false positive),  and 
let $\muNP$ be the mean time between unpredicted faults (false negative).
Finally, we define the mean time between events as $\munew$ (including all three event types).
The relationships between $\mu$, $\muP$, $\muNP$, and $\munew$ are the following:
\begin{itemize}
\item Rate of unpredicted faults: $\frac{1}{\muNP} = \frac{{1-\recall
    }}{\mu} $, since $1-\recall$ is the fraction of faults that are
  unpredicted;
\item Rate of predicted faults: $ \frac{\recall}{\mu} =
  \frac{\precision}{\muP}$, since $\recall$ is the fraction of faults
  that are predicted, and $\precision$ is the fraction of fault
  predictions that are correct;

\item Rate of events:
  $\frac{1}{\munew}=\frac{1}{\muP}+\frac{1}{\muNP}$, since events are
  either predictions (true or false), or unpredicted faults.
\end{itemize}

\subsection{Objective: waste minimization}

The natural objective is to minimize the expectation of the total execution time, 
 \emph{makespan}, of the application. Instead, in order to ease
mathematical derivations, we aim at minimizing the \emph{waste}.
The waste is the expected percentage of time lost, or “wasted”, during
the execution.  In other words, the \emph{waste} is the
fraction of time during which the platform is not doing useful work.  This
definition was introduced by
Wingstrom~\cite{wingstrom-phd}. Obviously,
the lower the waste, the lower the expected makespan, and
reciprocally. Hence the two objectives are strongly related and
minimizing one of them also minimizes the other. 

\section{Revisiting Daly's  first-order approximation}
\label{sec.youngdaly}

Young proposed in~\cite{young74} a ``first order approximation to the
optimum checkpoint interval''. Young's formula was later refined by
Daly~\cite{daly04} to take into account the recovery time. 
We revisit their analysis using the notion of
waste. 

Let  \Time[base] be the base time of the application without any
overhead (neither checkpoints nor faults).
First, assume a \emph{fault-free} execution of the application with periodic checkpointing.
In such an environment, during each period of
length $\period$ we take a checkpoint, which lasts for a time $\Cr$,
and only $\period-\Cr$ units of work are executed. Let \Time[FF] be
the execution time of the application in this setting. 
Following most works in the literature, we also take a checkpoint at
the end of the execution. The fault-free execution time \Time[FF] is equal to the time
needed to execute the whole application, \Time[base], plus the time taken by
the checkpoints:
\begin{equation}
\Time[FF] \ = \ \Time[base] + \Nckpt \Cr
\label{eq.timeff11}
\end{equation}
where \Nckpt is the number of checkpoints taken. We have 
$$\Nckpt = \left\lceil\frac{\Time[base]}{\period-\Cr}\right\rceil
\ \approx \frac{\Time[base]}{\period-\Cr}$$
When discarding the ceiling function, we assume that the execution time is
very large with respect to the period or, symmetrically, that there are many
periods during the execution. Plugging back the (approximated) value
$\Nckpt = \frac{\Time[base]}{\period-\Cr}$, we derive that
\begin{equation}
\Time[FF] \ = \frac{\Time[base]}{\period-\Cr} \period
\label{eq.timeff1}
\end{equation}

The waste due to checkpointing in a fault-free execution, \Waste[FF],
is defined as the fraction of the execution time that does not
contribute to the progress of the application:
\begin{equation}
\Waste[FF] = \frac{\Time[FF]-\Time[base]}{\Time[FF]} 
\qquad \Leftrightarrow \qquad
\big( 1 -\Waste[FF] \big) \Time[FF] = \Time[base]
\label{eq.ff}
\end{equation}
Combining Equations~\eqref{eq.timeff1} and~\eqref{eq.ff}, we get:
\begin{equation}
\Waste[FF] = \frac{\Cr}{\period}
\label{eq.wasteff}
\end{equation}

Now, let \Time[final] denote the expected execution time of
the application in the presence of faults. This execution
time can be divided into two parts: (i) the execution of ``chunks'' of
work of size $\period-\Cr$ followed by their checkpoint; and (ii) the time
lost due to the faults. This decomposition is illustrated by
Figure~\ref{fig.daly}. The first part of the execution time is equal
to \Time[FF]. Let \Nfaults be the number of faults occurring during
the execution, and let \Tlost be the average time lost per
fault. Then,
\begin{equation}
  \Time[final] = \Time[FF] + \Nfaults \times \Tlost
\label{eq:tfinalraw}
\end{equation}
On average, during a time $\Time[final]$, $\Nfaults=\frac{\Time[final]}{\mu}$
faults happen. We need to estimate \Tlost. The instants at
which periods begin and at which faults strike are
independent. Therefore, the expected time elapsed between the
completion of the last checkpoint and a fault is $\frac{\period}{2}$
for all distribution laws, regardless of their particular shape. We
conclude that  $\Tlost = \frac{\period}{2}+D+R$,
because after each fault there is a downtime and a recovery. 
This leads to:
$$
  \Time[final] = \Time[FF] + \frac{\Time[final]}{\mu} \times \left(\D+\R+ \frac{\period}{2} \right)
$$
Let \Waste[fault] be the fraction of the total execution time that is
lost because of faults: 
\begin{equation}
\Waste[fault] =
\frac{\Time[final]-\Time[FF]}{\Time[final]}
\qquad\Leftrightarrow\qquad
\left( 1 -\Waste[fault] \right) \Time[final] = \Time[FF]
\label{eq.fail}
\end{equation}
We derive:
\begin{equation}
  \Waste[fault] = \frac{1}{\mu} \left( \D + \R + \frac{\period}{2} \right).
  \label{eq.wastefailfail}
\end{equation}

\begin{figure}[h]
\centering
\input{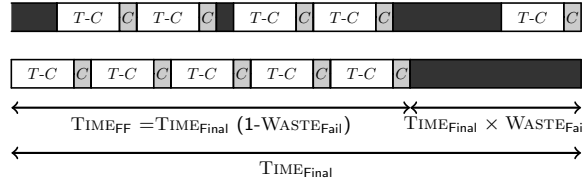}
\caption{An execution (top), and its re-ordering (bottom), to illustrate both sources
  of waste. Blackened intervals correspond to work destroyed by
  faults, downtimes, and recoveries.}
\label{fig.daly}
\end{figure}

In~\cite{daly04}, Daly uses the expression
\begin{equation}
	\Time[final] = \big( 1+ \Waste[fault] \big) \Time[FF]
\label{eq.dalydaly}
\end{equation}
instead of Equation~\eqref{eq.fail}, which leads him to his well-known first-order formula
\begin{equation}
\period = \sqrt{2(\mu + (\D+\R)) \Cr}+\Cr
\label{daly.wrong}
\end{equation}
Figure~\ref{fig.daly} explains why Equation~\eqref{eq.dalydaly} is not
correct and should be replaced by Equation~\eqref{eq.fail}. Indeed,
the expected number of faults depends on the final time, not on the
time for a fault-free execution. We point out that
Young~\cite{young74} also used Equation~\eqref{eq.dalydaly}, but with
$\D=\R=0$. Equation~\eqref{eq.fail} can be rewritten $\Time[final] =
\Time[FF]/\left( 1 -\Waste[fault] \right)$. Therefore, using
Equation~\eqref{eq.dalydaly} instead of Equation~\eqref{eq.fail}, in
fact, is equivalent to write
$
\frac{1}{1-\Waste[fault]} \approx 1+\Waste[fault]
$
which is indeed a first-order approximation if $\Waste[fault] \ll 1$.

Now, let \Waste denote the total waste:
\begin{equation}
\Waste= \frac{\Time[final]-\Time[base]}{\Time[final]}
\label{eq.wastefinal}
\end{equation}
Therefore
$$
\Waste = 1 - \frac{\Time[base]}{\Time[final]}
= 1 - \frac{\Time[base]}{\Time[FF]}\frac{\Time[FF]}{\Time[final]}
= 1 -(1-\Waste[FF])(1-\Waste[fault]).
$$
Altogether, we derive the final result:
\begin{eqnarray}
	\Waste &  = & \Waste[FF] + \Waste[fault] - \Waste[FF] \Waste[fault]
\label{eq.waste.right.gen}\\
	& = & \frac{\Cr}{\period} + \left( 1 -  \frac{\Cr}{\period} \right) 
	\frac{1}{\mu} \left( \D + \R + \frac{\period}{2} \right)
\label{eq.waste.right}
\end{eqnarray}

We obtain $\Waste =  \frac{u}{\period} + v + w \period$ where $u=\Cr \big( 1 - \frac{\D+\R}{\mu} \big)$, 
$v =  \frac{\D+\R- \Cr/2}{\mu} $, and $w= \frac{1}{2\mu}$. Thus $\Waste$ is minimized for $\period = \sqrt{\frac{u}{w}}$. 
The Refined First-Order (\newdaly) formula for the optimal period is thus:
\begin{equation}
\Tyafo = \sqrt{2(\mu - (\D+\R)) \Cr}
\label{eq.daly.right}
\end{equation}

It is interesting to point out why Equation~\eqref{eq.daly.right} is a first-order approximation, even for large jobs. 
Indeed, there are several restrictions to enforce for the approach to be valid:
\begin{compactitem}
\item We have stated that
the expected number of faults during execution is $\Nfaults = \frac{\Time[final]}{\mu} $,
and that the expected time lost due to a fault
is $\Tlost = \frac{\period}{2}$. Both statements are true individually,
but the expectation of a product is the product of the expectations only if the random variables are independent,
which is not the case here because \Time[final] depends upon the failure inter-arrival times.
\item In Equation~\eqref{eq.wasteff}, we have to enforce $\Cr \leq \period$ to have $\Waste[FF] \leq 1$
\item In Equation~\eqref{eq.wastefailfail}, we have to enforce $\D+\R \leq \mu$ and to bound  $\period$ in order to have $\Waste[fault] \leq 1$.
Intuitively, we need $\mu$ to be large enough for Equation~\eqref{eq.wastefailfail} to make sense. However,
regardless of the value of the individual MTBF $\mu_{\text{ind}}$, there is always a threshold in the number of components $N$
above which the platform MTBF  $\mu = \frac{\mu_{\text{ind}}}{N}$ becomes too small for Equation~\eqref{eq.wastefailfail} to be valid.
\item Equation~\eqref{eq.wastefailfail} is accurate only when 
two or more faults do not take place within the same period. Although unlikely when $\mu$ is large in front of $\period$,
the possible occurrence of many faults during the same period cannot be eliminated.
\end{compactitem}

To ensure that the latter condition (at most a single fault per period) is met with a high probability, we cap the length of the period: 
we enforce the condition $\T \leq \alpha \mu$, where $\alpha$ is some
tuning parameter chosen as follows. 
The number of faults during a period of length $\T$ can be modeled as a Poisson process of parameter 
$\beta = \frac{\T}{\mu}$.
The probability of having $k \geq 0$ faults is $P(X=k) = \frac{\beta^{k}}{k!}  e^{-\beta}$, where $X$ 
is the number of faults. 
Hence the probability of having two or more faults is $\pi = P(X\geq2) = 1 -( P(X=0) + P(X=1)) = 1 - (1+\beta) e^{-\beta}$.
If we assume $\alpha=0.27$ then $\pi \leq 0.03$, 
hence a valid approximation when bounding the period range accordingly. Indeed,
with such a  conservative value for $\alpha$, 
we have overlapping faults for only $3\%$ of the checkpointing segments in average, 
so that the model is quite reliable.
For consistency, we also enforce the same type of bound on the checkpoint time, and on the downtime and recovery: 
$\Cr \leq \alpha \mu$ and $\D+\R \leq \alpha \mu$.
However, enforcing these constraints may lead to use a sub-optimal period: it may well be the case that
 the optimal period $\sqrt{2(\mu - (\D+\R)) \Cr}$ of Equation~\eqref{eq.daly.right} does not belong to the 
 admissible interval $[\Cr , \alpha \mu]$. In that case, the waste is minimized for one of the bounds of the admissible interval:
 this is because, as seen from  Equation~\eqref{eq.waste.right},  the waste is a convex function of the period.

We conclude this discussion on a positive note. While capping the period, and enforcing a lower bound on the MTBF,
is mandatory for mathematical rigor, simulations (see Section~\ref{sec.simulations} for both Exponential 
and Weibull distributions) show that actual job executions can always use the value from 
 Equation~\eqref{eq.daly.right}, accounting for multiple faults
 whenever they occur by re-executing the work until success. The first-order model turns out to be surprisingly robust!

To the best of our knowledge, despite all the limitations above, there
is no better approach to estimate the waste due to checkpointing when
dealing with arbitrary fault distributions. However, assuming that
faults obey an Exponential distribution, it is possible to use the
memory-less property of this distribution to provide more accurate
results.  A second-order approximation when faults obey an
Exponential distribution is given in Daly~\cite[Equation~(20)]{daly04}
as $\Time[final] = \mu e^{\R/\mu} (e^{\frac{T}{\mu}}-1) \frac{
  \Time[base]}{\period -\Cr}$. In fact, in that case, the exact value
of $\Time[final]$ is provided in~\cite{SC2011,c183} as $\Time[final] = (\mu +
\D) e^{\R/\mu} (e^{\frac{T}{\mu}}-1) \frac{ \Time[base]}{\period
  -\Cr}$, and the optimal period is then
$\frac{1+\mathbb{L}(-e^{-\frac{\Cr}{\mu}-1})}{\mu}$ where
$\mathbb{L}$, the Lambert function, is defined as
$\mathbb{L}(z)e^{\mathbb{L}(z)}=z$.

To assess the accuracy of the different first order approximations, we
compare the periods defined by Young's formula~\cite{young74}, Daly's
formula~\cite{daly04}, and Equation~\eqref{eq.daly.right}, to the
optimal period, in the case of an Exponential distribution. Results
are reported in Table~\ref{tab.periods}.  To establish these results,
we use the same parameters as in Section~\ref{sec.simulations}: \Cr=
\R= 600 s, \D= 60 s, and $\mu_{\text{ind}} = 125$ years. Furthermore, to
compute the optimal period, for each platform size we choose the
application size so that $\Time[base]=2 $
hours. 
One can observe in Table~\ref{tab.periods} that the relative error for
Daly's period is slightly larger than the one for Young's period. In
turn, the absolute value of the relative error for Young's period is
slightly larger than the one for \newdaly. More importantly, when
Young's and Daly's formulas overestimate the period, \newdaly
underestimates it. Table~\ref{tab.periods} does not allow us to assess
whether these differences are actually significant. However we also 
report in Section~\ref{sec.simulations.synthetic} some simulations that
show that Equation~\eqref{eq.daly.right} leads to smaller execution
times for Weibull distributions than both classical formulas
(Tables~\ref{makespan.07.tab} and~\ref{makespan.05.tab}).

\begin{table}[h]
  \centering
  \begin{tabular}{|r|r|rr|rr|rr|r|}
    \hline
    $N$ & \multicolumn{1}{c|}{$\mu$} & \multicolumn{2}{c|}{\young} &
    \multicolumn{2}{c|}{\daly} & \multicolumn{2}{c|}{\newdaly}  & \multicolumn{1}{c|}{Optimal}
    \\\hline
$2^{10}$ & 	3849609 &	68567 & (0.5 \%)	 &	 68573 & (0.5 \%)	 &	 67961 & (-0.4 \%)	 &	68240	\\
$2^{11}$ & 	1924805 &	48660 & (0.7 \%)	 &	 48668 & (0.7 \%)	 &	 48052 & (-0.6 \%)	 &	48320	\\
$2^{12}$ & 	962402 &	34584 & (1.2 \%)	 &	 34595 & (1.2 \%)	 &	 33972 & (-0.6 \%)	 &	34189	\\
$2^{13}$ & 	481201 &	24630 & (1.6 \%)	 &	 24646 & (1.7 \%)	 &	 24014 & (-0.9 \%)	 &	24231	\\
$2^{14}$ & 	240601 &	17592 & (2.3 \%)	 &	 17615 & (2.5 \%)	 &	 16968 & (-1.3 \%)	 &	17194	\\
$2^{15}$ & 	120300 &	12615 & (3.2 \%)	 &	 12648 & (3.5 \%)	 &	 11982 & (-1.9 \%)	 &	12218	\\
$2^{16}$ & 	60150 &	9096 & (4.5 \%)	 &	 9142 & (5.1 \%)	 &	 8449 & (-2.9 \%)	 &	8701	\\
$2^{17}$ & 	30075 &	6608 & (6.3 \%)	 &	 6673 & (7.4 \%)	 &	 5941 & (-4.4 \%)	 &	6214	\\
$2^{18}$ & 	15038 &	4848 & (8.8 \%)	 &	 4940 & (10.8 \%)	 &	 4154 & (-6.8 \%)	 &	4458	\\
$2^{19}$ & 	7519 &	3604 & (12.0 \%)	 &	 3733 & (16.0 \%)	 &	 2869 & (-10.8 \%)	 &	3218	\\
\hline
  \end{tabular}
  \caption{Comparing periods produced by the
    different approximations with optimal value. Beside each period, we report its
    relative deviation to the optimal. Each value is expressed in seconds.\label{tab.periods}}
\end{table}

\section{Taking predictions into accounts}
\label{sec.no.intervals}

In this section, we present an analytical model to assess the impact of predictions on periodic checkpointing strategies.
As already mentioned, we consider the case
where the predictor is able to provide exact prediction dates, 
and to generate such predictions at least $\Cp$ seconds in advance, so that a proactive checkpoint of length $\Cp$
can indeed be taken before  the event.

For the sake of clarity, we start with a simple algorithm (Section~\ref{sec.algo.simple}) which we refine
in Section~\ref{sec.algo.beautiful}. We then compute the value of the period that minimizes the waste in  
Section~\ref{sec.waste.minimization}.

\subsection{Simple policy}
\label{sec.algo.simple}

In this section, we consider the following algorithm:
\begin{itemize}
\item While no fault prediction is available, checkpoints are taken
  periodically with period $\period$;
  \item When a fault is predicted, there are two cases: either there is the possibility to take a proactive checkpoint,
  or there is not enough time
to do so, because we are already checkpointing (see Figures~\ref{fig.waste-exact}(b) and~\ref{fig.waste-exact}(c)).
In the latter case, there is no other choice than ignoring the prediction. 
In the former case, we still have the possibility to ignore the prediction, but we may also decide to trust it: in fact
 the decision is randomly taken. With
  probability \trust, we trust the predictor and take the prediction
  into account  (see Figures~\ref{fig.waste-exact}(f) and~\ref{fig.waste-exact}(g)),
  and with probability $1-\trust$,  we ignore the
  prediction  (see Figures~\ref{fig.waste-exact}(d) and~\ref{fig.waste-exact}(e));
    \item If we take the prediction into account, we take a proactive checkpoint (of length $\Cp$)  as late as possible, i.e., so that it completes right
 at the time when the fault is predicted to happen. After this checkpoint, we  complete the execution of the period  
 (see Figures~\ref{fig.waste-exact}(f) and~\ref{fig.waste-exact}(g));
  
   \item If we ignore the prediction, either by necessity (not enough time to take an extra checkpoint, see 
   Figures~\ref{fig.waste-exact}(b) and~\ref{fig.waste-exact}(c)), or
    or by choice (with 
 probability $1-\trust$, Figures~\ref{fig.waste-exact}(d) and~\ref{fig.waste-exact}(e)), 
 we finish the current period and start a new one.
\end{itemize}

The rationale for not always trusting the predictor is to avoid taking
useless checkpoints too frequently. Intuitively, the precision $\precision$
of the predictor must be above a given threshold for its usage to be worthwhile.
In other words, if we decide to checkpoint just before a predicted event,
either we will save time by avoiding a costly re-execution if the event does correspond
to an actual fault, or  we will lose time by unduly performing an extra checkpoint.
We need a larger proportion of the former
cases, i.e., a good precision, for the predictor to be really useful.
The following analysis will determine the optimal value of  $\trust$ as a function
of the parameters $\Cr$, $\Cp$, $\mu$, $\recall$, and $\precision$.

We could refine the approach by taking  into account  the amount of work already done in the current period
when deciding whether to trust the predictor or not. Intuitively, the more work already done, the more important
to save it, hence the more worthwhile to trust the predictor. We design such a refined strategy in Section~\ref{sec.algo.beautiful}.
Right now, we analyze a simpler algorithm where we decide to trust or not to trust the predictor, independently of the amount of 
work done so far within the period.
 

\begin{figure*}
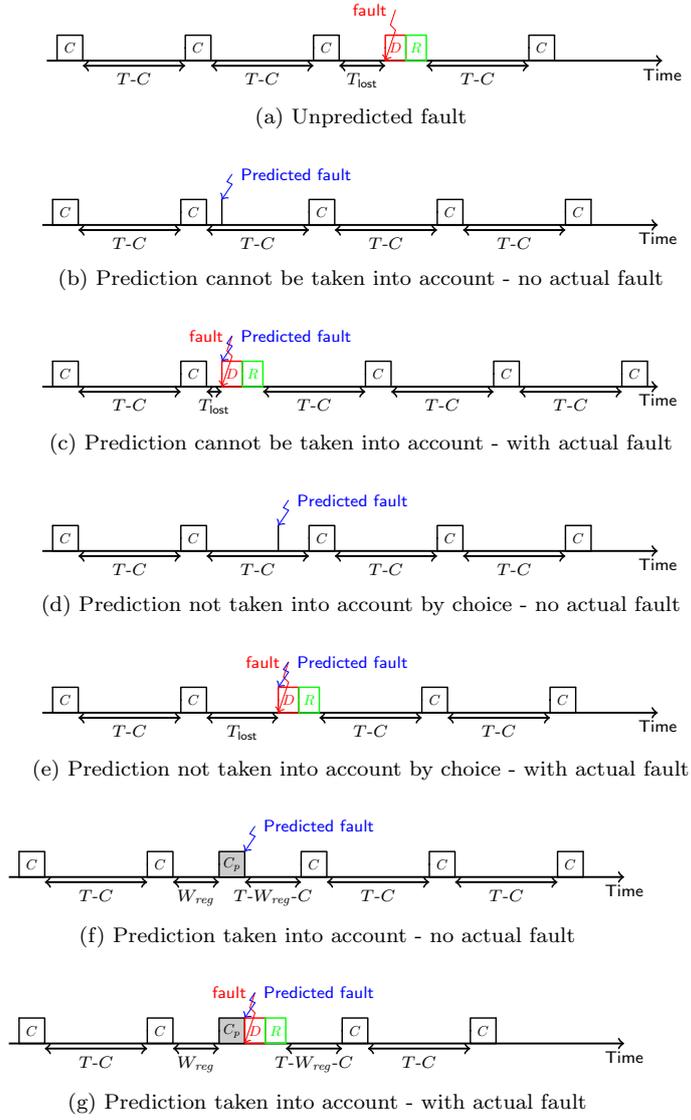

\centering
\subfloat[Unpredicted fault]
{
\scalebox{0.9}{\input{unpred_faults.tex}}
}
\\
\subfloat[Prediction cannot be taken into account - no actual fault]
{
\scalebox{0.9}{\input{pred_cannot_taken_nofault.tex}}
}
\\
\subfloat[Prediction cannot be taken into account - with actual fault]
{
\scalebox{0.9}{\input{pred_cannot_taken_wfault.tex}}
}
\\
\subfloat[Prediction not taken into account by choice - no actual fault]
{
\scalebox{0.9}{\input{pred_not_taken_nofault.tex}}
}
\\
\subfloat[Prediction not taken into account by choice - with actual fault]
{
\scalebox{0.9}{\input{pred_not_taken_wfault.tex}}
}
\\
\hspace{-1cm}
\subfloat[Prediction taken into account - no actual fault]
{
\scalebox{0.9}{\input{pred_taken_nofault.tex}}
}
\\
\hspace{-1cm}
\subfloat[Prediction taken into account - with actual fault]
{
\scalebox{0.9}{\input{pred_taken_wfault.tex}}
}
\caption{Actions taken for the different event types.}
\label{fig.waste-exact}
\end{figure*}

We analyze the algorithm in order to compute a formula for the expected waste, just as
in Equation~\eqref{eq.waste.right}. While the value of \Waste[FF] is unchanged ($\Waste[FF] = \frac{\Cr}{\T}$), 
the value of \Waste[fault] is modified because of predictions. As
illustrated in Figure~\ref{fig.waste-exact}, 
there are many different scenarios that contribute to \Waste[fault] that can be sorted into three categories:

\smallskip
\noindent
(1) \textbf{Unpredicted faults:}
 This overhead occurs each time an unpredicted fault strikes, that is, on average, once every $\muNP$ seconds.      
 Just as in Equation~\eqref{eq.wastefailfail}, the corresponding waste is
  $\frac{1}{\muNP} \left[ \frac{\T}{2} + \D+  \R  \right]$.\\
  
  \smallskip
\noindent
  (2) \textbf{Predictions not taken into account:}
The second source of waste is for predictions that are ignored. This overhead occurs in two different scenarios.
First, if we do not have time to take a proactive checkpoint, we have
an overhead if and only the prediction is an actual fault. This case
happens with probability $\precision$. We then lose a time $t+\D+\R$
if the predicted fault happens a time $t$ after the
completion of the last periodic checkpoint. The expected time lost is thus
$$\Tlost^{1} = \frac{ 1}{\T} \int_0^{\Cp} \left( \p(t+\D+\R) + (1-\p) 0 \right) dt$$
Then, if we do have  time to take a proactive checkpoint but still decide to ignore the prediction, 
we also have an overhead if and only the prediction is an actual fault, but the expected time lost is now weighted by
the probability $(1-\trust)$:
$$\Tlost^{2} = (1-\trust)  \frac{ 1}{\T} \int_{\Cp}^{\T} \left( \p(t+\D+\R) + (1-\p) 0 \right)dt $$

\smallskip
\noindent
(3) \textbf{Predictions taken into account:}
We now compute the overhead due to a prediction which we trust (hence we checkpoint just before its date). 
If the prediction is an actual fault, we lose $\Cp+\D+\R$  seconds, but if it is not, we lose the 
unnecessary extra checkpoint time $\Cp$. The expected time lost is now weighted by
the probability \trust and becomes
 $$\Tlost^{3} = \trust \frac{ 1}{\T} \int_{\Cp}^{\T} \left( \p(\Cp+\D+\R) + (1-\p) \Cp \right)dt $$

We derive the final value of \Waste[fault]:
$$\Waste[fault]  = \frac{1}{\muNP} \left[ \frac{\T}{2} + \D+  \R  \right] + \frac{1}{\muP} \left[ \Tlost^{1} + \Tlost^{2} + \Tlost^{3} \right]$$
This final expression comes from the disjunction of all possibles cases, using the Law of Total 
Probability~\cite[p.23]{Mitzenmacher2005}: 
the waste comes either from non-predicted faults or from predictions; in the latter case, we have analyzed the three
possible sub-cases and weighted them with their respective probabilities.
 After simplifications, we obtain
\begin{equation}
\Waste[fault] = \frac{1}{\mu}\left ((1- \recall \trust) \frac{\T}{2} + \D+ \R + \frac{\trust \recall}{\p} \Cp - \frac{\trust \recall \Cp^2}{\precision\T} (1 - \precision/2) \right )
\label{eq.total-waste}
\end{equation}
We could now plug this expression back into Equation~\eqref{eq.waste.right.gen} to compute the value of $\period$ 
that minimizes the total waste.
Instead,  we move on to describing the refined algorithm, and we minimize the waste for the refined strategy, since it always
induces a smaller waste.

\subsection{Refined policy}
\label{sec.algo.beautiful}

In this section, we refine the approach and  consider different trust 
strategies, depending upon the time in the period where the prediction takes place. Intuitively, the later in the period,
the more likely we are inclined to trust the predictor, because the amount of work that we could lose gets larger and larger.
As before, we cannot take into account a fault predicted to happen less
than $\Cp$ units of time after the beginning of the period. Therefore,
we focus on what happens in the period after time $\Cp$.
Formally, we now divide the interval $[\Cp, \period]$ into $n$ intervals $[\beta_i;\beta_{i+1}]$ for 
$i \in \{0,\cdots,n-1\}$, where $\beta_0=\Cp$ and $\beta_n=\T$. For each interval 
$[\beta_i;\beta_{i+1}]$, we trust the predictor with probability $\trust_i$. We aim at determining the values of 
$n$,  $\beta_i$, and $\trust_{i}$ that minimize the waste. 
As mentioned before, intuition tells us that the $\trust_{i}$ values should be non-decreasing. 
We prove below a somewhat unexpected theorem: in the optimal strategy, there is either one or two different $\trust_{i}$
values, and these values are $0$ or $1$. This means that we should \emph{never} trust the predictor in the beginning of a period,
and always trust it in the end of the period, without any intermediate behavior in between.

We formally express this striking result below. 
Let $\beta_{\text{lim}}=\frac{\Cp}{\precision}$.
The optimal strategy is provided by Theorem~\ref{th.thopt} below. We first prove the following proposition:

\begin{proposition}
\label{th.gopi}
The values of $\beta_i$ and $\trust_{i}$ that minimize the waste
satisfy the following conditions:\\
(i) For all $i$ such that $\beta_{i+1} \leq \beta_{\text{lim}}$, $\trust_i=0$.\\
(ii) For all $i$ such that $\beta_{i} \geq \beta_{\text{lim}}$, $\trust_{i}=1$. 
\end{proposition}

\begin{proof}
First we compute the waste with the refined algorithm, using Equation~\eqref{eq.waste.right.gen}. 
The formula for \Waste[fault] is similar to Equation~\eqref{eq.total-waste} on each interval:

\begin{align*}
\Waste &= \frac{\Cr}{\T} + \left (1 - \frac{\Cr}{\T} \right) \left [ \frac{1}{\muNP}\left ( \frac{\T}{2} + \D + \R\right) \right .\\
& +\frac{1}{\muP}\sum_{i=0}^{n-1}  \left( \trust_i  \int_{\beta_i}^{\beta_{i+1}} \frac{ (\precision(\Cp + \D+\R) + (1-\precision)\Cp)}{\T} dt  \right . \\
&  +\left . \left . (1-\trust_i) \int_{\beta_i}^{\beta_{i+1}} \frac{\precision(t+\D+\R)}{\T}dt \right)  \right] 
\end{align*}

Now, consider a fixed value of $i$ and express the value of $\Waste$ as a function of $\trust_{i}$:
\begin{align*}
\Waste
&= K +\left ( 1 - \frac{\Cr}{\T} \right ) \frac{\trust_i}{\muP} \int_{\beta_i}^{\beta_{i+1}} \left ( \frac{ \Cp}{\T} -\frac{\precision t}{\T}\right )dt
\end{align*}
where $K$ does not depend on $\trust_i$. From the sign of the function to be integrated, 
one sees that  $\Waste$ is minimized 
when $\trust_i=0$ if $\beta_{i+1} \leq \beta_{\text{lim}}=\frac{\Cp}{p}$, and when $\trust_i=1$ if $\beta_{i}\geq \beta_{\text{lim}}$.
\end{proof}

\begin{theorem}
The optimal algorithm takes proactive actions if and only if the prediction falls in the interval 
$[\beta_{\text{lim}},\T]$.
\label{th.thopt}
\end{theorem}

\begin{proof}
From Proposition~\ref{th.gopi},  the values for $\trust_i$ are optimally defined for every $i$ but one: we do 
not know the optimal value if there exists $i_{0}$ such that $\beta_{i_{0}} < \beta_{\text{lim}} < \beta_{i_{0}+1}$.
Then let us consider the waste where $\trust_{i_{0}}$
is replaced by $\trust_{i_{0}}^{(1)}$ on $[\beta_{i_{0}},\beta_{\text{lim}}]$ and by $\trust_{i_{0}}^{(2)}$ 
on $[\beta_{\text{lim}}, \beta_{i_{0}+1}]$.
The new waste is necessarily smaller than the one with only $\trust_{i_0}$, since we relaxed the 
constraint. We know from Proposition~\ref{th.gopi} that the optimal
solution is then to have 
$\trust_{i_0}^{(1)}=0$ and $\trust_{i_0}^{(2)}=1$.
\end{proof}

Let us now compute the value of the waste with the optimal algorithm. There are two cases, depending upon whether
$\period \leq \beta_{\text{lim}}$ or not. For values of $\T$ smaller than  $\beta_{\text{lim}}$, Theorem~\ref{th.thopt} 
shows that the optimal algorithm never takes any proactive action; in that case the waste is given by Equation~\eqref{eq.waste.right}
in Section~\ref{sec.youngdaly}.
For values of $\T$ larger than  $\beta_{\text{lim}} = \frac{\Cp}{\precision}$, we compute the waste due to predictions as
\begin{align*}
&\frac{1}{\muP} \frac{1}{\T} \left( \int_0^{\Cp/\precision} \precision(t+\D+\R) dt 
+  \int_{\Cp/\precision}^\T ( \precision(\Cp + \D+\R) + (1-\precision)\Cp) dt \right) \\
& = \frac{\recall}{\precision \mu} \left( \precision(\D+\R) +\Cp - \frac{\Cp^{2}}{2 \precision \T} \right)
\end{align*}
Indeed, in accordance with Theorem~\ref{th.thopt}, no prediction is taken into account in the interval $[0,\frac{\Cp}{\precision}]$, while all predictions 
are taken into account in the interval $[\frac{\Cp}{\precision},\period]$. 
Adding the waste due to unpredicted faults, namely $\frac{1}{\muNP} \left[ \frac{\T}{2} + \D+  \R  \right]$,
we derive
\[\Waste[fault] =  \frac{1}{\mu} \left ( (1-\recall) \frac{\T}{2}  + 
\frac{\recall}{\precision} \Cp \left ( 1 - \frac{1}{2\precision}\frac{\Cp}{\T}\right ) + \D+\R \right ).
\]
Plugging this value into Equation~\eqref{eq.waste.right.gen}, we obtain the total waste when $\frac{\Cp}{\precision}\leq \T$:
\begin{align*}
\waste 
& = \frac{\Cr}{\T} + \frac{1}{\mu} \left ( (1-\recall) \frac{\T}{2}  + 
\frac{\recall}{\precision} \Cp \left ( 1 - \frac{1}{2\precision}\frac{\Cp}{\T}\right ) + \D+\R \right )\left ( 1-\frac{\Cr}{\T}\right ) \nonumber \\
& = \frac{\recall\Cr\Cp^2}{2\precision^2}\frac{1}{\mu \T^2} + \left(\mu \Cr - \frac{\recall\Cp^2}{2\precision^2} - \Cr\left(\frac{\recall\Cp}{\precision}+\D+\R\right)\right)\frac{1}{\mu\T} +  \frac{1-\recall}{2\mu}\T \nonumber \\
& \quad \quad+ \frac{-(1-\recall)\frac{\Cr}{2} + \frac{\recall \Cp}{\precision} + \D + \R}{\mu}
\end{align*}

Altogether, the expression for the total waste becomes:
\begin{equation}
\begin{cases}
\waste_{1}(T) = \frac{\Cr \left ( 1 - \frac{\D+\R}{\mu} \right)}{\period} + \frac{\D+\R - \Cr/2}{\mu} + \frac{1}{2\mu} \period  & \text{if } \frac{\Cp}{\precision}\geq \T \\[4mm]
\waste_{2}(T) = \frac{\recall\Cr\Cp^2}{2\mu\precision^2}\frac{1}{ \T^2} + \frac{\left(\Cr \left ( 1 - \frac{\frac{\recall\Cp}{\precision}+\D+\R}{\mu} \right) - \frac{\recall\Cp^2}{2\mu\precision^2} \right)}{\T} + \frac{-(1-\recall)\frac{\Cr}{2} + \frac{\recall \Cp}{\precision} + \D + \R}{\mu} +  \frac{1-\recall}{2\mu}\T  & \text{if } \frac{\Cp}{\precision}\leq \T \\
\end{cases}
\label{eq.gopi-final}
\end{equation}
One can check that when $\recall=0$ (no error predicted, hence no proactive action in the algorithm), 
then $\waste_{1}$ and $\waste_{2}$  coincide. We also check that both values coincide for $\T = \frac{\Cp}{\precision}$.
We show how to minimize the waste in Equation~\eqref{eq.gopi-final} in Section~\ref{sec.waste.minimization}.

\subsection{Waste minimization}
\label{sec.waste.minimization}

In this section we focus on minimizing the waste in Equation~\eqref{eq.gopi-final}.
Recall that, by construction, we always have to enforce the constraint $\T \geq \Cr$.
First consider the case where $\Cr \leq \frac{\Cp}{\precision}$. 
On the interval $\T \in [\Cr,\frac{\Cp}{\precision}]$, we retrieve the optimal value found 
in Section~\ref{sec.youngdaly}, and derive that $\waste_{1}$, the
waste when predictions are not taken into account, is minimized for
\begin{equation}
	\label{eq.ty.beautiful}
\Ty=\max \left ( \Cr, \min \left (\Tyafo, \frac{\Cp}{\precision} \right ) \right )
\end{equation}
Indeed, the optimal value should belong to the interval $[\Cr,\frac{\Cp}{\precision}]$, and the 
function $\waste_{1}$ is convex: if the extremal solution $\sqrt{2(\mu-(\D+\R))\Cr}$ does not belong to this interval,
then the optimal value is one of the bounds of the interval.

On the interval  $\T \in \left[\frac{\Cp}{\precision},+\infty\right)$, we find the optimal solution by differentiating
twice $\waste_{2}$ with respect to \T.  Writing $\waste_{2}(\T) = \frac{u}{\T^{2}} + \frac{v}{\T} + w + x \T$ for simplicity,
we obtain $\waste_{2}''(\T)  = \frac{2}{\T^{3}} \left( \frac{3u}{\T} + v \right)$. 
Here, a key parameter is 
the sign of : $$v = \left(\Cr \left ( 1 - \frac{\frac{\recall\Cp}{\precision}+\D+\R}{\mu} \right) - \frac{\recall\Cp^2}{2\mu\precision^2} \right)$$
 We detail the case $v\geq 0$ in the following, because it is the most frequent with realistic
parameter sets; we do have $v\geq 0$ for all the whole range of simulations in Section~\ref{sec.simulations}.
For the sake of completeness,  we will briefly discuss the case $v<0$ in the comments below.

When $v \geq 0$, we have $\waste_{2}''(\T) \geq 0$, so that $\waste_{2}$ is convex on the interval
$\left[\frac{\Cp}{\precision},+\infty \right)$ and admits a unique minimum \Te. Note that \Te can be computed either numerically or using Cardano's method, since it is the unique real root of a polynomial of degree $3$.
%
The optimal solution on $\left[\frac{\Cp}{\precision},+\infty\right)$ is then:
$
\Tp=\max \left ( \Te, \frac{\Cp}{\precision} \right )
$.

It remains to consider the case  where  $\frac{\Cp}{\precision}<\Cr$. In fact, it suffices to add the constraint that the value of 
\Tp should be greater than \Cr, that is:
\begin{equation}
	\label{eq.tp.beautiful}
\Tp=\max \left ( \Cr, \max \left (\Te, \frac{\Cp}{\precision} \right ) \right )
\end{equation}
Finally, the optimal solution for the waste is given by the minimum of the following two values:
\[
\begin{cases}
\frac{\Cr \left ( 1 - \frac{\D+\R}{\mu} \right)}{\Ty} + \frac{\D+\R - \Cr/2}{\mu} + \frac{1}{2\mu} \Ty \\[4mm]
\frac{\recall\Cr\Cp^2}{2\mu\precision^2}\frac{1}{ \Tp^2} + \frac{\left(\Cr \left ( 1 - \frac{\frac{\recall\Cp}{\precision}+\D+\R}{\mu} \right) - \frac{\recall\Cp^2}{2\mu\precision^2} \right)}{\Tp} + \frac{-(1-\recall)\frac{\Cr}{2} + \frac{\recall \Cp}{\precision} + \D + \R}{\mu} +  \frac{1-\recall}{2\mu}\Tp
\end{cases}
\]

We make a few observations:
\begin{itemize}
\item Just as for Equation~\eqref{eq.daly.right} in Section~\ref{sec.youngdaly}, mathematical rigor calls for capping the
values of \D, \R, \Cr, \Cp and \T in front of the MTBF. The only difference is that we should replace $\mu$ by
$\munew$: this is to account for the occurrence rate of all events, be they unpredicted faults or predictions.
\item While the expression of the waste looks complicated, the numerical value of the optimal period
can easily be computed in all cases. We have dealt with the case $v \geq 0$, 
where $v$ is the coefficient of $1/T$ in $\waste_{2}(\T) = \frac{u}{\T^{2}} + \frac{v}{\T} + w + x \T$.
When $v < 0$ we only needs to compute all the nonnegative real roots
of a polynomial of degree $3$, and check which one leads to the best value. More precisely, these root(s) partition the
admissible interval $\left[\frac{\Cp}{\precision},+\infty \right)$ into several sub-intervals, and the optimal value is either a root or a sub-interval bound.
\item In many practical situations, when $\mu$ is large enough, we can dramatically simplify
the expression of $\waste_{2}(\T)$: we have $\T = O(\sqrt{\mu})$,
 the term $ \frac{u}{\T^{2}}$ becomes negligible, checkpoint parameters become negligible in front of $\mu$,
 and we derive the approximated value $\sqrt{\frac{2 \mu \Cr}{1-\recall}}$. This value can be seen as an extension of
Equation~\eqref{eq.daly.right} giving \Tyafo, where $\mu$ is replaced by $\frac{\mu}{1-\recall}$: faults are replaced by non-predicted faults, 
 and the overhead due to false predictions is negligible. As a word of caution, recall that this conclusion is valid only when $\mu$
 is very large in front of all other parameters.
\end{itemize}


\section{Simulation results}
\label{sec.simulations}

We start by presenting the simulation framework
(Section~\ref{sec.simulations.framework}). Then we report results
using synthetic traces (Section~\ref{sec.simulations.synthetic}) and
log-based traces (Section~\ref{sec:logbased}). Finally, we assess the
respective impact of the two key parameters of a predictor, its recall
and its precision, on checkpointing strategies
(Section~\ref{section.impact}).

\subsection{Simulation framework}
\label{sec.simulations.framework}

\noindent\textbf{Scenario generation --}
In order to check the accuracy of our model and of our analysis, and to
assess the potential benefits of predictors, we study the
performance of our new solutions and of pre-existing ones using a
discrete-event simulator.  The simulation engine generates a random
trace of faults.  Given a set of $p$ processors, a failure trace is a
set of failure dates for each processor over a fixed time horizon $h$
(set to 2 years).  Given the distribution of inter-arrival times at a
processor, for each processor we generate a trace via independent
sampling until the target time horizon is reached.  The job start time
is assumed to be one-year to avoid side-effects related to the
synchronous initialization of all nodes/processors. We consider two
types of failure traces, namely synthetic and log-based.

\noindent\textbf{Synthetic failure traces --} 
The simulation engine generates a random trace of faults
parameterized either by an Exponential fault distribution or by
Weibull distribution laws with shape parameter either $0.5$ or $0.7$.  Note
that Exponential faults are widely used for theoretical studies, while
Weibull faults are representative of the behavior of real-world
platforms~\cite{Weibull1,Weibull2,liu2008optimal,Heien:2011:MTH:2063384.2063444}. 
For example, Heien et al.~\cite{Heien:2011:MTH:2063384.2063444}
have studied the failure distribution for 6 sources of failures (storage devices, NFS, batch system, memory and processor cache errors, etc.), and the aggregate failure distribution. They have shown that the aggregate failure distribution is best modeled by a Weibull distribution with a shape parameter that is between 0.5841 and 0.7097.

The Jaguar platform, which comprised $N=45,208$ processors, is
reported to have experienced about one fault per day~\cite{6264677},
which leads to an individual (processor) MTBF $\mu_{\text{ind}}$ equal
to $\frac{45,208}{365}\approx 125$ years. Therefore, we set the
individual (processor) MTBF to $\mu_{\text{ind}} = 125$ years. We let
the total number of processors $N$ vary from $N=16,384$ to
$N=524,288$, so that the platform MTBF $\mu$ varies from $\mu=4,010$
min (about $2.8$ days) down to $\mu=125$ min (about $2$ hours).
Whatever the underlying failure distribution, it is scaled so that its
expectation corresponds to the platform MTBF $\mu$.  The application
size is set to $\Time[base]=10,000$ years/N.

\noindent
\textbf{Log-based failure traces --}  To corroborate the results obtained
with synthetic failure traces, and to further assess the performance
of our algorithms, we also perform simulations using the failure logs of two
production clusters.  We use logs of the largest clusters among the
preprocessed logs in the \emph{Failure trace
  archive}~\cite{10.1109/CCGRID.2010.71}, i.e., for clusters at the
Los Alamos National Laboratory~\cite{Weibull2}. In these logs, each
failure is tagged by the node ---and not the processor--- on which the
failure occurred. Among the 26 possible clusters, we opted for the
logs of the only two clusters with more than 1,000 nodes. The
motivation is that we need a sample history sufficiently large to
simulate platforms with more than ten thousand nodes.  The two chosen
logs are for clusters 18 (LANL18) and 19 (LANL19) in the archive
(referred to as 7 and 8 in~\cite{Weibull2}).  For each log, we record
the set $\mathcal{S}$ of availability intervals.  The discrete failure
distribution for the simulation is generated as follows: the
conditional probability $\mathbb{P}(X \ge t \ |\ X \ge \tau)$ that a
node stays up for a duration $t$, knowing that it has been up for a
duration $\tau$, is set to the ratio of the number of availability
durations in $\mathcal{S}$ greater than or equal to $t$, over the
number of availability durations in $\mathcal{S}$ greater than or
equal to $\tau$.

The two clusters used for computing our log-based failure
distributions consist of 4-processor nodes. Hence, to simulate a
platform of, say, $2^{16}$ processors, we generate $2^{14}$ failure
traces, one for each 4-processor node.  In the logs the individual
(processor) MTBF is $\mu_{\text{ind}}=691$ days for the LANL18
cluster, and $\mu_{\text{ind}}=679$ days for the LANL19 cluster.  The
LANL18 and LANL19 traces are logs for systems which comprised 4,096
processors. Using these logs to generate traces for a system made of
$524,288$ processors, as the largest platforms we consider with
synthetic failure traces, would lead to an obvious risk of
oversampling. Therefore, we limit the size of the log-based traces we
generate: we let the total number of processors $N$ varies from
$N=1,024$ to $N=131,072$, so that the platform MTBF $\mu$ varies from
$\mu=971$ min (about $16$ hours) down to $\mu=7.5$ min.  The application
size is set to $\Time[base]=250$ years/N.

\noindent 
\textbf{Predicted failures and false predictions --}  Once we have
generated a failure trace, we need to determine which faults are
predicted and which are not. In order to do so, we consider all faults in a trace
one by one.  For each of them, we randomly decide, with probability
\recall, whether it is predicted.

We use the simulation engine to generate a random trace of false
predictions. The main problem is to decide the shape of the
distribution that false predictions should follow. To the best of our
knowledge, no published study ever addressed that problem. For
synthetic failure traces, we report results when false predictions
follow the same distribution than faults (except, of course, that both
distributions do not have the same mean
value). In~\ref{app.sup.plots}, we report on simulations when false
predictions are generated according to a uniform distribution; the
results are quite similar. For log-based failures, we only report
results when false predictions are generated according to a uniform
distribution (because we believe that scaling down a discrete, actual
distribution may not be meaningful).

The distribution of false predictions is always scaled so that its
expectation is equal to $\frac{\muP}{1-\precision} = \frac{\precision
  \mu}{\recall (1-\precision)}$, the inter-arrival time of false
predictions.  Finally, the failure trace and the false-prediction
trace are merged to produce the final trace including all events (true
predictions, false predictions, and non predicted faults). Each
reported value is the average over $100$ randomly generated instances.

\noindent
\textbf{Checkpointing, recovery, and downtime costs --}  The experiments
use parameters that are representative of current and forthcoming
large-scale platforms~\cite{j116,Ferreira2011}.  We take $\Cr=\R=10$
min, and $\D=1$ min for the synthetic failure traces. For the
log-based traces we consider smaller platforms. Therefore, we take
$\Cr=\R=1$ min, and $\D=6$s.  Whatever the trace, we consider three
scenarios for the proactive checkpoints: either proactive checkpoints
are (i) exactly as expensive as periodic ones ($\Cp=\Cr$), (ii) ten
times cheaper ($\Cp=0.1\Cr$), and (iii) two times more expensive
($\Cp=2\Cr$).

\noindent
\textbf{Heuristics --}
In the simulations, we compare four checkpointing strategies:
\begin{compactitem}
\item \newdaly is the checkpointing strategy of period 
$\period = \sqrt{2(\mu - (\D+\R)) \Cr}$  (see Section~\ref{sec.youngdaly}).
\item \OptimalPrediction is the refined algorithm described in Section~\ref{sec.algo.beautiful}.

\item To assess the quality of each strategy, we compare it with its
  \bestper counterpart, defined as the same strategy but using the
  best possible period $\T$. This latter period is computed via a
  brute-force numerical search for the optimal period (each tested
  period is evaluated on $100$ randomly generated traces, and the period
  achieving the best average performance is elected as the ``best
  period''). 
\end{compactitem}

\noindent\textbf{Fault predictors -- }
We experiment using the characteristics of two predictors from the
literature: one accurate predictor with high recall and
precision~\cite{5958823}, namely with $\precision=0.82$ and
$\recall=0.85$, and another predictor with intermediate recall and
precision~\cite{5542627}, namely with $\precision=0.4$ and
$\recall=0.7$.

In practice, a predictor will not be able to predict the exact time
at which a predicted fault will strike the system. Therefore, in the
simulations, when a predictor predicts that a failure will strike the
system at a date $t$ (true prediction), the failure actually occurs
exactly at time $t$ for heuristic \OptimalPrediction, and between
time $t$ and time $t+2\Cr$ for heuristic \inexact (the probability of
fault is uniformly distributed in the time-interval). \OptimalPrediction
can thus be seen as a best case. The comparison between
\OptimalPrediction and \inexact enables us to assess the impact of the
time imprecision of predictions, and to show that the obtained results are
quite robust to this type of imprecision. The choice of an interval
length of $2\Cr$ is quite arbitrary. For synthetic traces, this
corresponds to 1,200 s, which is quite a significant imprecision.

\subsection{Simulations with synthetic traces}
\label{sec.simulations.synthetic}

Figures~\ref{fig.082.085} and~\ref{fig.04.07} show the average waste
degradation for the two checkpointing policies, and for their \bestper
counterparts, for both predictors. The waste is reported as a function
of the number of processors $N$. We draw the plots as a function of
the number of processors $N$ rather than of the platform MTBF $\mu =
\mu_{\text{ind}}/N$, because it is more natural to see the waste
increase with larger platforms. However, recall that this work is
agnostic of the granularity of the processing elements and
intrinsically focuses on the impact of the MTBF on the waste.

We also report job execution times, in Table~\ref{makespan.exp.tab}
when fault distribution follows an Exponential distribution law, and
in Tables~\ref{makespan.07.tab} and \ref{makespan.05.tab} for a
Weibull distribution law with shape parameter $k=0.7$ and $k=0.5$
respectively.

\noindent\textbf{Validation of the theoretical study --}
We used Maple to analytically compute and plot the optimal value of
the waste for both the algorithm taking predictions into account,
\OptimalPrediction, and for the algorithm ignoring them, \newdaly.  In
order to check the accuracy of our model, we have compared these
results with results obtained with the discrete-event simulator.

We first observe that there is a very good correspondence between
analytical results and simulations in Figures~\ref{fig.082.085}
and~\ref{fig.04.07}. In particular, the Maple plots and the
simulations for Exponentially distributed faults are very
similar. This shows the validity of the model and of its analysis.
Another striking result is that \OptimalPrediction has the same waste
as its \bestper counterpart, even for Weibull fault distributions, in
all but the most extreme cases. In the other cases, the waste achieved
by \OptimalPrediction is very close to that of its \bestper
counterpart. This demonstrates the very good quality of our
checkpointing period \Tp.  These conclusions are valid regardless of
the cost ratio of periodic and proactive checkpoints.  

In Tables~\ref{makespan.exp.tab} through \ref{makespan.05.tab} we
report the execution times obtained when using the expression of \T
given by Young~\cite{young74} and Daly~\cite{daly04} (denoted
respectively as \young and \daly) to assess whether \Tyafo is a better
approximation.  (Recall that these three approaches ignore the
predictions, which explains why the numbers are identical on both
sides of each table.)  The expressions of \T given by \young, \daly,
and \newdaly are identical for Exponential distributions and the three
heuristics achieve the same performance
(Table~\ref{makespan.exp.tab}).  This confirms the analytical
evaluation of Table~\ref{tab.periods} in Section~\ref{sec.youngdaly}.
For Weibull distributions (Tables~\ref{makespan.07.tab} and
\ref{makespan.05.tab}), \newdaly achieves lower makespan, and the
difference becomes even more significant as the size of the platform
increases.  Moreover, it is striking to observe in
Table~\ref{makespan.05.tab} that job execution time increases together
with the number for processors (from $N=2^{16}$ to $N=2^{19}$) if the
checkpointing period is \daly or \young.  On the contrary, job
execution time (rightfully) decreases when using \newdaly, even if the
decrease is moderate with respect to the increase of the platform
size.  Altogether, the main (striking) conclusion is that \newdaly
should be preferred to both classical approaches for Weibull
distributions.

\noindent\textbf{The benefits of prediction --}
The second observation is that the prediction is useful for the vast
majority of the set of parameters under study!  In addition, when
proactive checkpoints are cheaper than periodic ones, the benefits of
fault prediction are increased.  On the contrary, when proactive
checkpoints are more expensive than periodic ones, the benefits of
fault prediction are greatly reduced. One can even observe that the
waste with prediction is not better than without prediction in the
following scenario: $\Cp=2\Cr$, and using the limited-quality
predictor ($\precision=0.4$, $\recall=0.7$) with $2^{19}$ processors,
see Figures~\ref{fig.04.07}(i),(j),(k), and~(l).


%




\begin{figure*}
\centering
\includegraphics[scale=0.5]{newlegend.tex}\\
\hspace{-1cm}
\subfloat[Maple]
{
\includegraphics[scale=0.16]{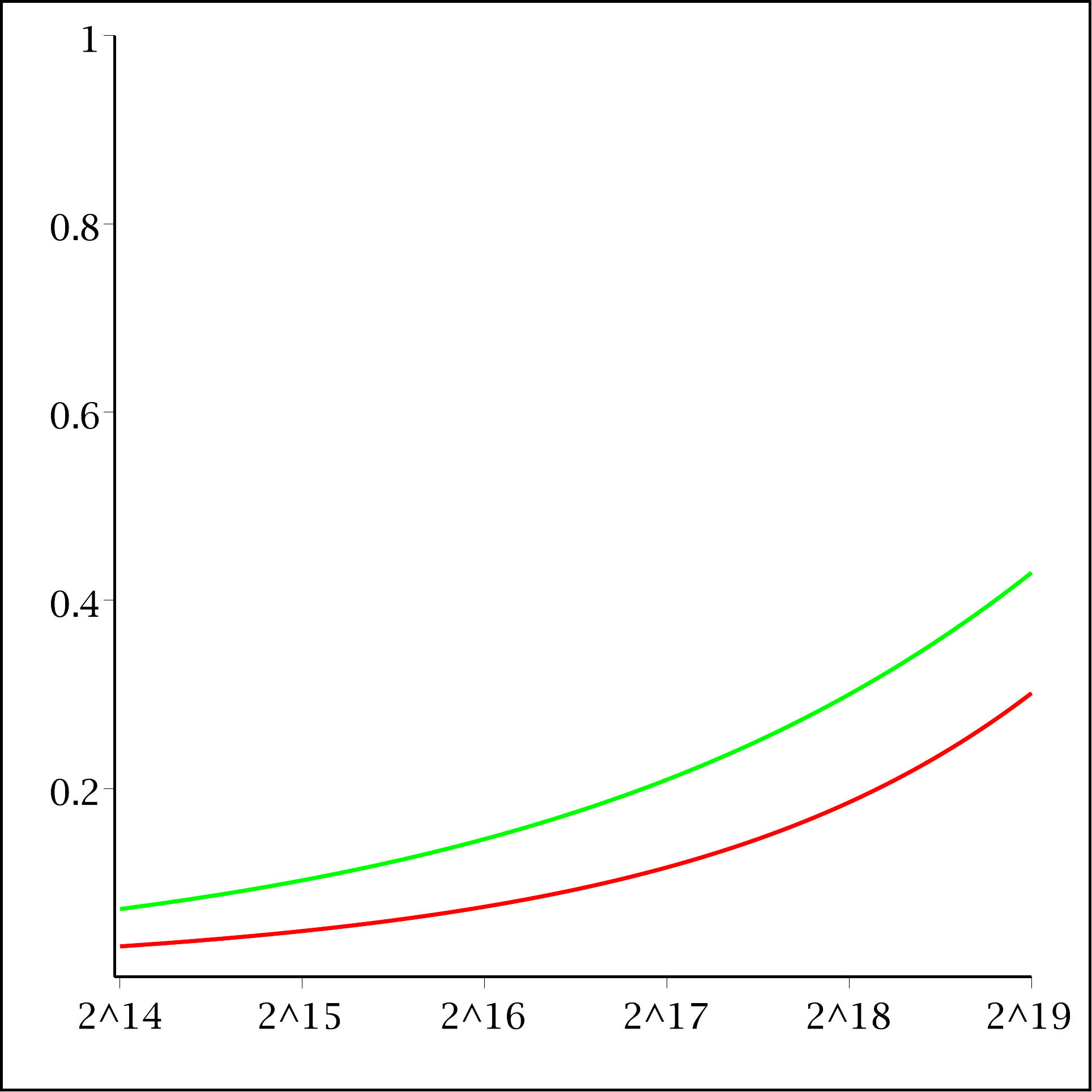}
}
\subfloat[Exponential]
{
\includegraphics[scale=0.36]{r085p083-EXP-fixedC-appli0-platform-variation.tex}
}	
\subfloat[Weibull $k=0.7$]
{
\includegraphics[scale=0.36]{r085p083-WEIBULL-07-fixedC-appli0-platform-variation.tex}
}
\subfloat[Weibull $k=0.5$]
{
\includegraphics[scale=0.36]{r085p083-WEIBULL-05-fixedC-appli0-platform-variation.tex}
}
\\
\hspace{-1cm}
\subfloat[Maple]
{
\includegraphics[scale=0.16]{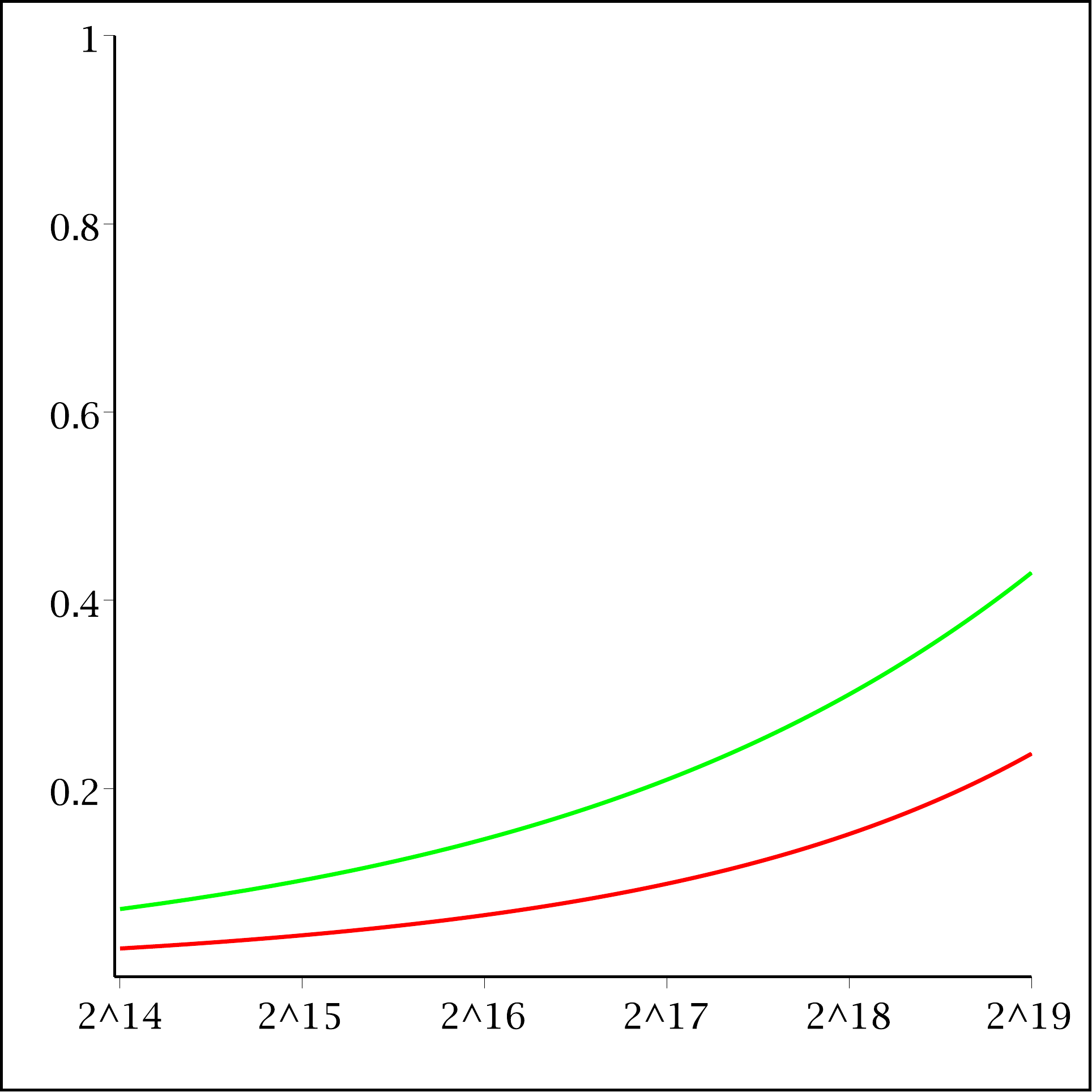}
}
\subfloat[Exponential]
{
\includegraphics[scale=0.36]{r085p083alpha1sur10-EXP-fixedC-appli0-platform-variation.tex}
}	
\subfloat[Weibull $k=0.7$]
{
\includegraphics[scale=0.36]{r085p083aplha1sur10-WEIBULL-07-fixedC-appli0-platform-variation.tex}
}
\subfloat[Weibull $k=0.5$]
{
\includegraphics[scale=0.36]{r085p083aplha1sur10-WEIBULL-05-fixedC-appli0-platform-variation.tex}
}
\\
\hspace{-1cm}
\subfloat[Maple]
{
\includegraphics[scale=0.16]{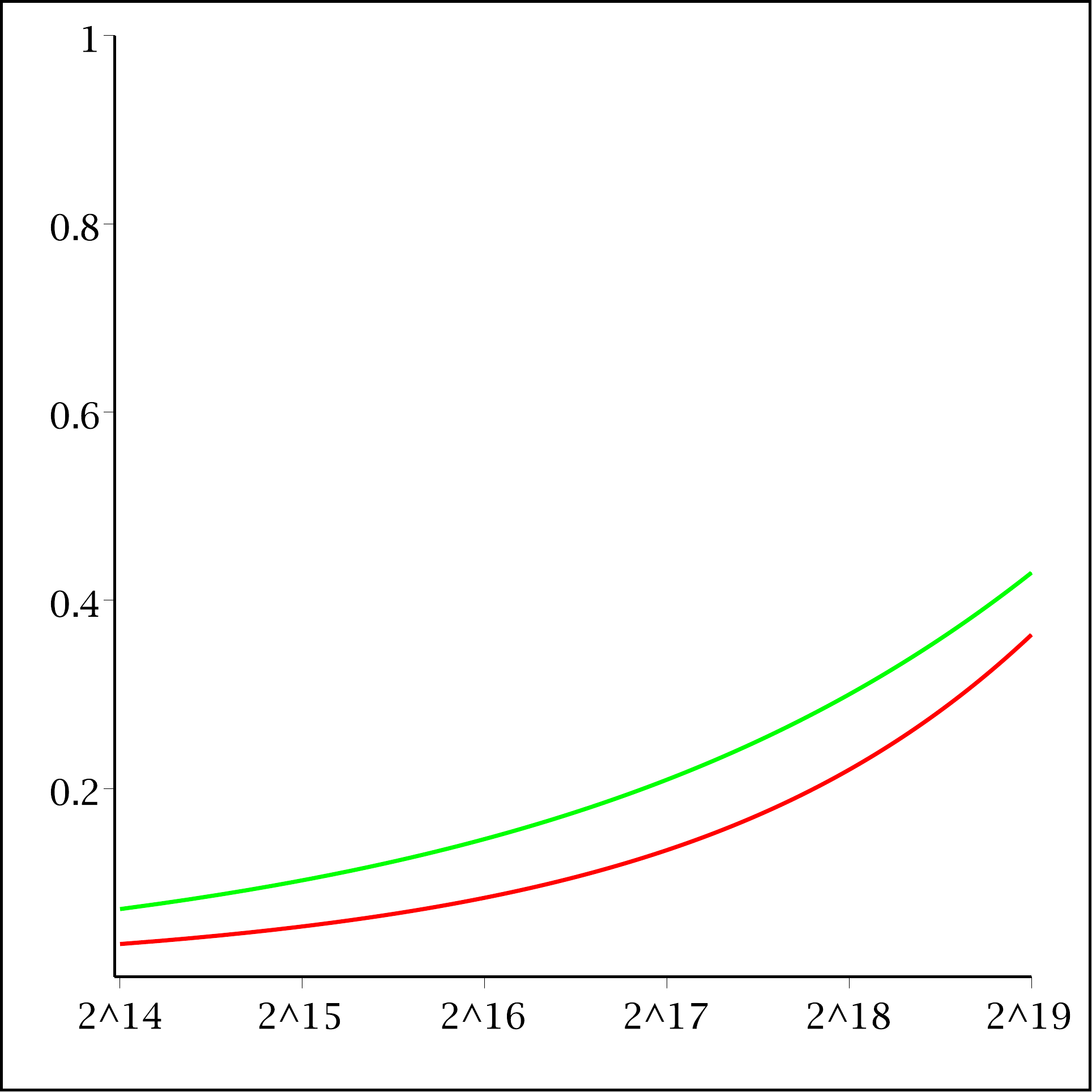}
}
\subfloat[Exponential]
{
\includegraphics[scale=0.36]{r085p083alpha2-EXP-fixedC-appli0-platform-variation.tex}
}	
\subfloat[Weibull $k=0.7$]
{
\includegraphics[scale=0.36]{r085p083alpha2-WEIBULL-07-fixedC-appli0-platform-variation.tex}
}
\subfloat[Weibull $k=0.5$]
{
\includegraphics[scale=0.36]{r085p083aplha2-WEIBULL-05-fixedC-appli0-platform-variation.tex}
}
\caption{Waste (y-axis) for the different heuristics as a function of the
  platform size (x-axis), with $\precision=0.82$, $\recall=0.85$, $\Cp=\Cr$ (first row), $\Cp=0.1 \Cr$  (second row), or $\Cp=2 \Cr$  (third  row) and with a trace of false predictions parametrized by a distribution identical to the distribution of the failure trace.}
	\label{fig.082.085}
\end{figure*}


\begin{figure*}
\centering
\hspace{-1cm}
\subfloat[Maple]
{
\includegraphics[scale=0.16]{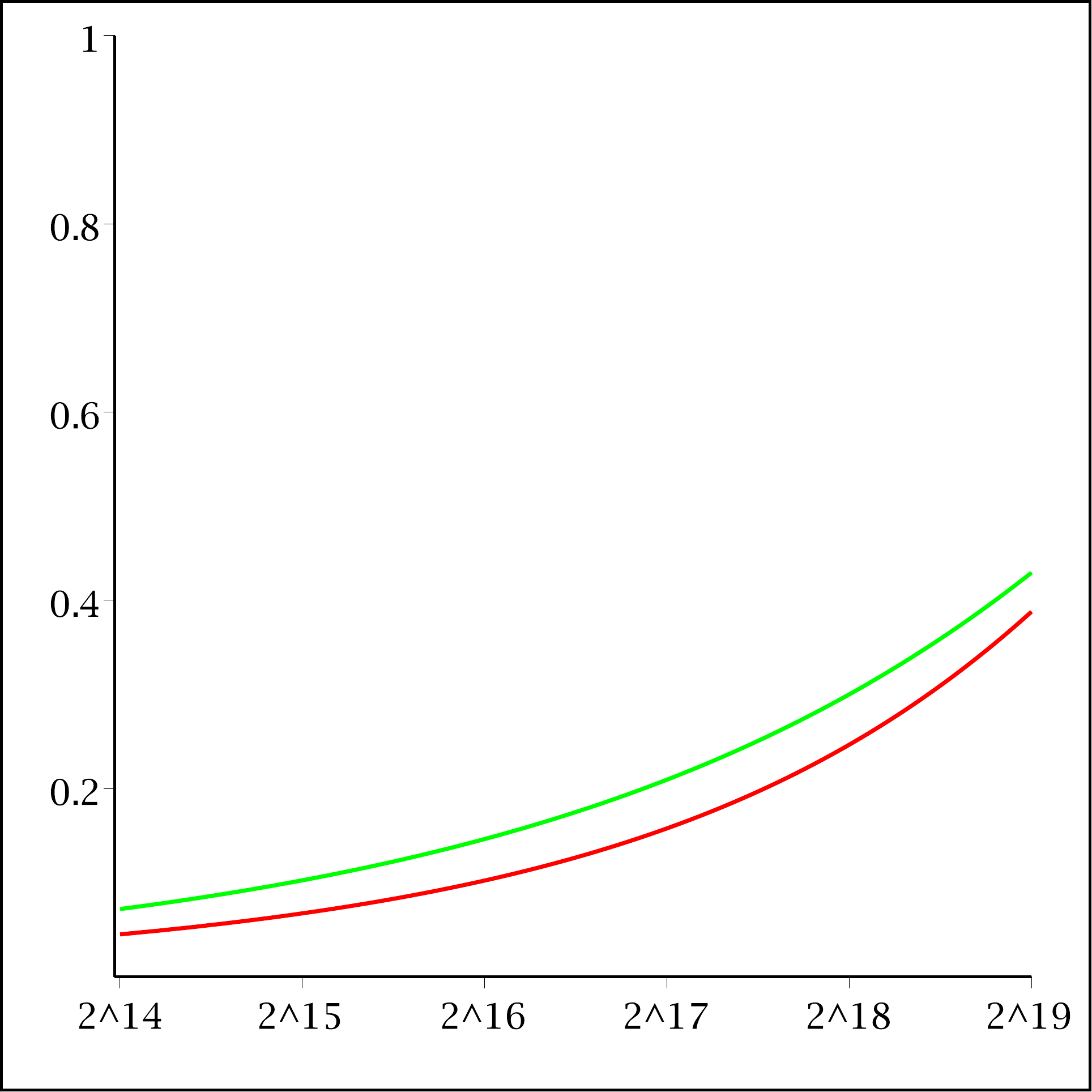}
}
\subfloat[Exponential]
{
\includegraphics[scale=0.36]{r07p04-EXP-fixedC-appli0-platform-variation.tex}
}
\subfloat[Weibull $k=0.7$]
{
\includegraphics[scale=0.36]{r07p04-WEIBULL-07-fixedC-appli0-platform-variation.tex}
}
\subfloat[Weibull $k=0.5$]
{
\includegraphics[scale=0.36]{r07p04-WEIBULL-05-fixedC-appli0-platform-variation.tex}
}
\\
\hspace{-1cm}
\subfloat[Maple]
{
\includegraphics[scale=0.16]{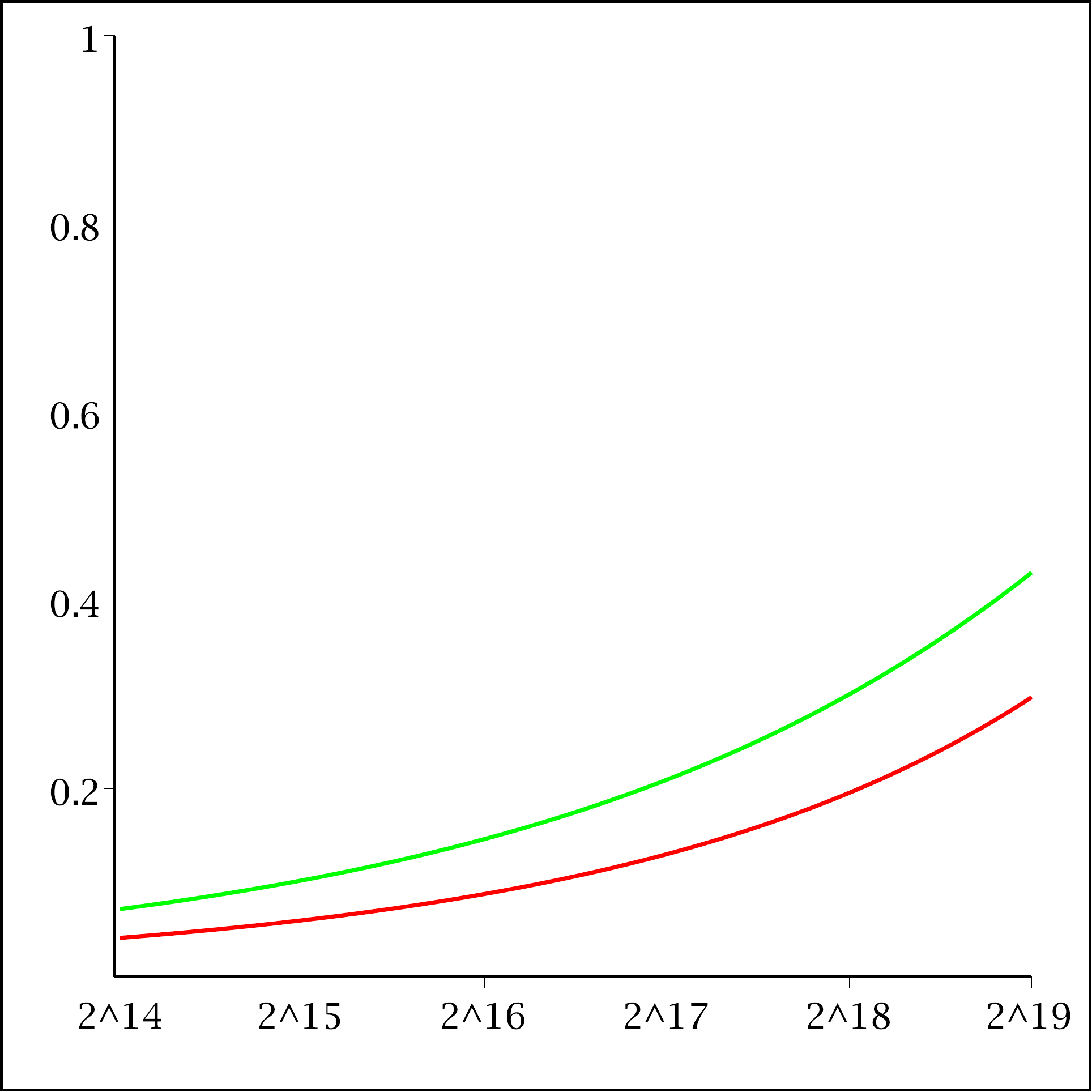}
}
\subfloat[Exponential]
{
\includegraphics[scale=0.36]{r07p04aplha1sur10-EXP-fixedC-appli0-platform-variation.tex}
}
\subfloat[Weibull $k=0.7$]
{
\includegraphics[scale=0.36]{r07p04alpha1sur10-WEIBULL-07-fixedC-appli0-platform-variation.tex}
}
\subfloat[Weibull $k=0.5$]
{
\includegraphics[scale=0.36]{r07p04alpha1sur10-WEIBULL-05-fixedC-appli0-platform-variation.tex}
}
\\
\hspace{-1cm}
\subfloat[Maple]
{
\includegraphics[scale=0.16]{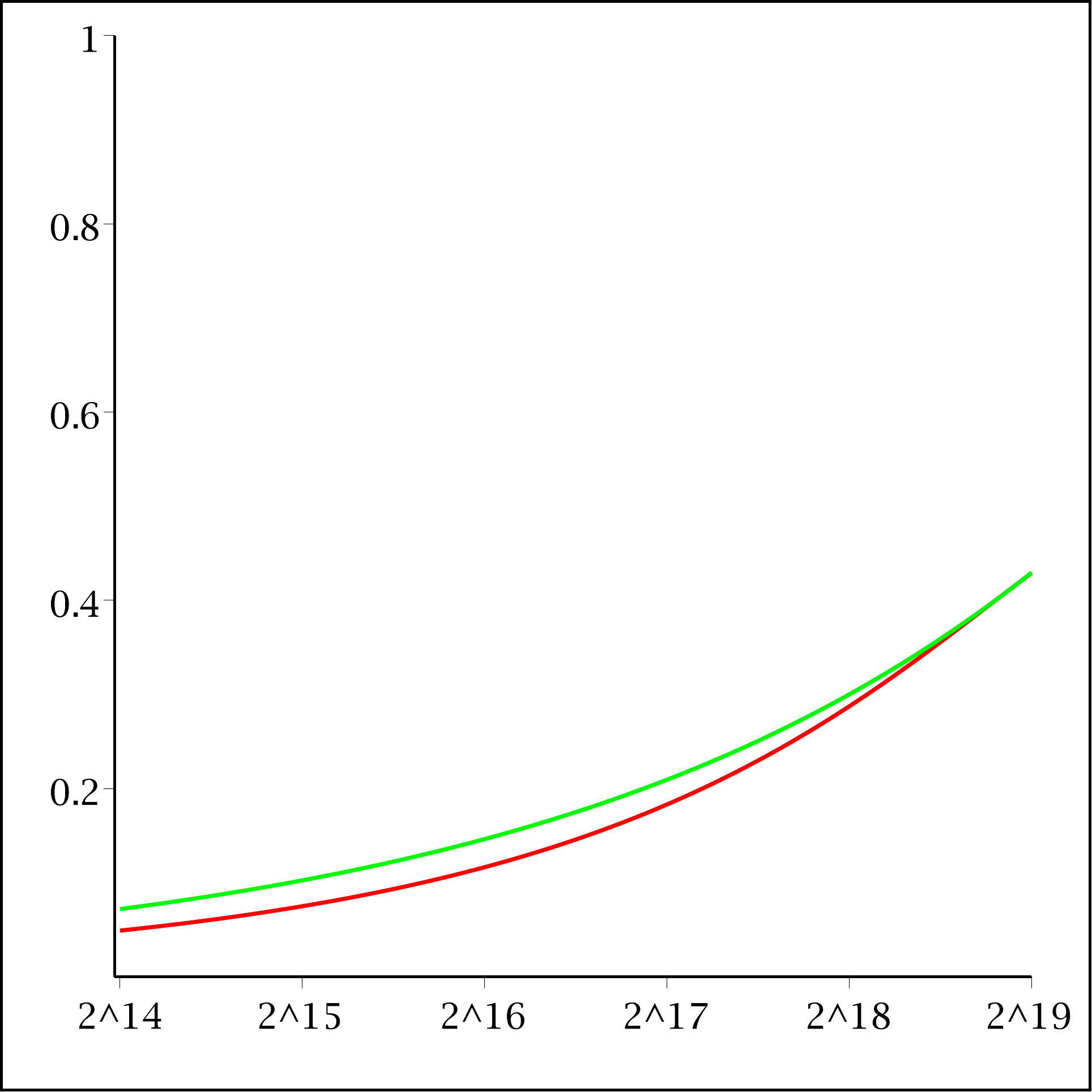}
}
\subfloat[Exponential]
{
\includegraphics[scale=0.36]{r07p04aplha2-EXP-fixedC-appli0-platform-variation.tex}
}
\subfloat[Weibull $k=0.7$]
{
\includegraphics[scale=0.36]{r07p04alpha2-WEIBULL-07-fixedC-appli0-platform-variation.tex}
}
\subfloat[Weibull $k=0.5$]
{
\includegraphics[scale=0.36]{r07p04alpha2-WEIBULL-05-fixedC-appli0-platform-variation.tex}
}
\caption{Waste (y-axis) for the different heuristics as a function of the
  platform size (x-axis), with $\precision=0.4$, $\recall=0.7$,  $\Cp=\Cr$  (first row), $\Cp=0.1 \Cr$ (second row), or $\Cp=2 \Cr$ (third  row) and with a trace of false predictions parametrized by a distribution identical to the distribution of the 
failure trace.}
	\label{fig.04.07}
\end{figure*}

In Tables~\ref{makespan.exp.tab} through \ref{makespan.05.tab} we
compute the gain (expressed in percentage) achieved by
\OptimalPrediction over \newdaly. As a general trend, we observe that
the gains due to predictions are more important when the distribution
law is further apart from an Exponential distribution. Indeed, the
largest gains are when the fault distribution follows a Weibull law of
parameter $0.5$.  Using \OptimalPrediction in conjunction with a
``good'' fault predictor we report gains up to 66\% when there is a large number
of processors ($2^{19}$). The gain is still of 37\% with $2^{16}$
processors. Using a predictor with limited recall and precision,
\OptimalPrediction can still decrease the execution time by $47\%$
with $2^{19}$ processors, and $31\%$ with $2^{16}$ processors. In all
tested cases, the decrease of the execution times is significant.
Gains are less important with Weibull laws of shape parameter $k=0.7$,
however they are still reaching a minimum of $13\%$ with $2^{16}$
processors, and up to $38\%$ with $2^{19}$ processors. Finally, gains
are further reduced with an Exponential law. They are still reaching
at least $5\%$ with $2^{16}$ processors, and up to $19\%$ with
$2^{19}$ processors.

The performance of \inexact shows that using a fault predictor remains
largely beneficial even in the presence of large uncertainties on the
time the predicted faults will actually occur (see
Tables~\ref{makespan.exp.tab},~\ref{makespan.07.tab},
and~\ref{makespan.05.tab}). When $N=2^{16}$ the degradation with
respect to \OptimalPrediction is of 3\% for a Weibull law with shape
parameter $k=0.7$, and the minimum gain over \newdaly is still of
10\%. When the shape parameter of the Weibull law is $k=0.5$, the
degradation is of 7\% when, for a minimum gain of 26\% over \newdaly.

\begin{table}[ht]
\centering
\begin{tabular}{c||c|c||c|c||}
 & \multicolumn{2}{|c||}{Execution time (in days)} & \multicolumn{2}{|c||}{Execution time (in days)}\\
\Cp = \Cr  & \multicolumn{2}{|c||}{($\precision = 0.82$,
  $\recall=0.85$)} & \multicolumn{2}{|c||}{($\precision = 0.4$, $\recall=0.7$)}\\
 & $2^{16}$ procs & $2^{19}$ procs & $2^{16}$ procs & $2^{19}$ procs \\
\hline
\young & 65.2 & 11.7 & 65.2 & 11.7 \\
\daly & 65.2 & 11.8 & 65.2 & 11.8 \\
\newdaly & 65.2 & 11.7 & 65.2  & 11.7\\
\hline
\OptimalPrediction & 60.0 (8\%) & 9.5 (19\%)& 61.7 (5\%) & 10.7 (8\%)\\
\inexact & 60.6 (7\%) & 10.2 (13\%)& 62.3 (4\%) & 11.4 (3\%) \\
\end{tabular}
\caption{Job execution times for an Exponential distribution, and gains due to the fault
  predictor (with respect to the performance of \newdaly).}
\label{makespan.exp.tab}
\end{table}

\begin{table}[ht]
\centering
\begin{tabular}{c||c|c||c|c||}
 & \multicolumn{2}{|c||}{Execution time (in days)} & \multicolumn{2}{|c||}{Execution time (in days)}\\
\Cp = \Cr  & \multicolumn{2}{|c||}{($\precision = 0.82$,
  $\recall=0.85$)} & \multicolumn{2}{|c||}{($\precision = 0.4$, $\recall=0.7$)}\\
 & $2^{16}$ procs & $2^{19}$ procs & $2^{16}$ procs & $2^{19}$ procs \\
\hline
\young & 81.3 & 30.1 & 81.3 &  30.1 \\
\daly & 81.4 & 31.0 & 81.4 & 31.0 \\
\newdaly & 80.3 & 25.5 & 80.3 & 25.5\\
\hline
\OptimalPrediction & 65.9 (18\%) & 15.9 (38\%)& 69.7 (13\%) & 20.2 (21\%)\\
\inexact & 68.0 (15\%) & 20.3 (20\%)& 72.0 (10\%) & 24.6 (4\%) \\
\end{tabular}
\caption{Job execution times for a Weibull distribution with
  shape parameter $k=0.7$, and gains due to the fault
  predictor (with respect to the performance of \newdaly).}
\label{makespan.07.tab}
\end{table}


\begin{table}[ht]
\centering
\begin{tabular}{c||c|c||c|c||}
 & \multicolumn{2}{|c||}{Execution time (in days)} & \multicolumn{2}{|c||}{Execution time (in days)}\\
\Cp = \Cr  & \multicolumn{2}{|c||}{($\precision = 0.82$,
  $\recall=0.85$)} & \multicolumn{2}{|c||}{($\precision = 0.4$, $\recall=0.7$)}\\
 & $2^{16}$ procs & $2^{19}$ procs & $2^{16}$ procs & $2^{19}$ procs \\
\hline
\young & 125.5 & 171.8 & 125.5 & 171.8 \\
\daly & 125.8 & 184.7 & 125.8 & 184.7 \\
\newdaly & 120.2 & 114.8 & 120.2  & 114.8\\
\hline
\OptimalPrediction & 75.9 (37\%) & 39.5 (66\%)& 83.0 (31\%) & 60.8 (47\%)\\
\inexact & 82.0 (32\%) & 60.8 (47\%)& 89.4 (26\%) & 76.6 (33\%) \\
\end{tabular}
\caption{Job execution times for a Weibull distribution with
  shape parameter $k=0.5$, and gains due to the fault
  predictor (with respect to the performance of \newdaly).}

\label{makespan.05.tab}
\end{table}

\subsection{Simulations with log-based traces}
\label{sec:logbased}

Figure~\ref{fig.traces18and19} shows the average waste degradation for
the two checkpointing policies, and for their \bestper counterparts,
for both predictors, both traces, and the three scenarios for
proactive checkpoints. Tables~\ref{makespan.trace18.tab}
and~\ref{makespan.trace19.tab} present job execution times for
\newdaly, \OptimalPrediction, and \inexact, for both traces and for
platform sizes smaller than as the ones reported in
Tables~\ref{makespan.exp.tab} through \ref{makespan.05.tab} for
synthetic traces. The waste for \newdaly is closer to its \bestper
counterpart with log-based traces than with Weibull-based traces. As a
consequence, when prediction with \OptimalPrediction is beneficial, it
is beneficial with respect to both \newdaly, and to \newdaly's
\bestper.

Overall, we observe similar results and reach the same
conclusions with log-based traces as with synthetic ones. The
waste of \OptimalPrediction is very close to that of its \bestper
counterpart for platforms containing up to $2^{16}$ processors.  This
demonstrates the validity of our analysis for the actual traces
considered. The waste of \OptimalPrediction is often significantly
larger than that of its \bestper counterpart for platforms containing
$2^{17}$ processors. The problem with the largest considered platforms
may be due to oversampling. Indeed, the original logs recorded events
for platforms comprising only 4,096 processors and respectively
contained only 3,010 and 2,343 availability intervals.

As with synthetic failure traces, prediction turns out to be useful
for the vast majority of tested configurations. The only cases when
prediction is not useful is with the ``bad'' predictor ($\recall=0.7$
and $\precision=0.4$), when the cost of proactive checkpoint is larger
than the cost of periodic checkpoints ($\Cp=2\Cr$), and when
considering the largest of platforms ($N=2^{17}$).  This extreme case
is, however, the only one for which prediction is not beneficial.  It
is not surprising that predictions are not useful when there are a lot
of false predictions that require the use of expensive proactive
actions. Looking at Tables~\ref{makespan.trace18.tab}
and~\ref{makespan.trace19.tab}, one could remark that performance
gains due to the predictions are similar to the ones observed with
Exponential-based traces, and are significantly smaller than the ones
observed with Weibull-based traces. However, recall
that we remarked that gains increase with the size of the platform,  and
that we consider smaller platforms when using log-based traces.

Finally, the imprecision related to the time where predicted faults
strike, induces a performance degradation. However, this degradation is
rather limited for the most efficient of the two predictors
considered, or when the platform size is not too large.

\begin{figure*}
\centering
\includegraphics[scale=0.5]{newlegend.tex}\\
\subfloat[$\Cp=0.1 \Cr$]
{
\includegraphics[scale=0.36]{TRACE18-R85-alpha01-fixedC-appli0-platform-variation.tex}
}
\subfloat[$\Cp=1 \Cr$]
{
\includegraphics[scale=0.36]{TRACE18-R85-alpha1-fixedC-appli0-platform-variation.tex}
}
\subfloat[$\Cp=2 \Cr$]
{
\includegraphics[scale=0.36]{TRACE18-R85-alpha2-fixedC-appli0-platform-variation.tex}
}\\
\centerline{LANL18 cluster with $\precision=0.82$, $\recall=0.85$.}
\smallskip \smallskip
\subfloat[$\Cp=0.1 \Cr$]
{
\includegraphics[scale=0.36]{TRACE18-R07-alpha01-fixedC-appli0-platform-variation.tex}
}
\subfloat[$\Cp=1 \Cr$]
{
\includegraphics[scale=0.36]{TRACE18-R07-alpha1-fixedC-appli0-platform-variation.tex}
}
\subfloat[$\Cp=2 \Cr$]
{
\includegraphics[scale=0.36]{TRACE18-R07-alpha2-fixedC-appli0-platform-variation.tex}
}\\
\centerline{LANL18 cluster with $\precision=0.4$, $\recall=0.7$.}
\smallskip\smallskip
\subfloat[$\Cp=0.1 \Cr$]
{
\includegraphics[scale=0.36]{TRACE19-R85-alpha01-fixedC-appli0-platform-variation.tex}
}
\subfloat[$\Cp=1 \Cr$]
{
\includegraphics[scale=0.36]{TRACE19-R85-alpha1-fixedC-appli0-platform-variation.tex}
}
\subfloat[$\Cp=2 \Cr$]
{
\includegraphics[scale=0.36]{TRACE19-R85-alpha2-fixedC-appli0-platform-variation.tex}
}\\
\centerline{LANL19 cluster with $\precision=0.82$, $\recall=0.85$.}
\smallskip \smallskip
\subfloat[$\Cp=0.1 \Cr$]
{
\includegraphics[scale=0.36]{TRACE19-R07-alpha01-fixedC-appli0-platform-variation.tex}
}
\subfloat[$\Cp=1 \Cr$]
{
\includegraphics[scale=0.36]{TRACE19-R07-alpha1-fixedC-appli0-platform-variation.tex}
}
\subfloat[$\Cp=2 \Cr$]
{
\includegraphics[scale=0.36]{TRACE19-R07-alpha2-fixedC-appli0-platform-variation.tex}
}\\
\centerline{LANL19 cluster  with $\precision=0.4$, $\recall=0.7$.}
\caption{Waste (y-axis) for the different heuristics as a function of the
  platform size (x-axis) with failures based on the failure log of
  LANL clusters 18 and 19.}
\label{fig.traces18and19}
\end{figure*}


\begin{table}[ht]
\centering
\begin{tabular}{c||c|c||c|c||}
 & \multicolumn{2}{|c||}{Execution time (in days)} & \multicolumn{2}{|c||}{Execution time (in days)}\\
\Cp = \Cr  & \multicolumn{2}{|c||}{($\precision = 0.82$,
  $\recall=0.85$)} & \multicolumn{2}{|c||}{($\precision = 0.4$, $\recall=0.7$)}\\
 & $2^{14}$ procs & $2^{17}$ procs & $2^{14}$ procs & $2^{17}$ procs \\
\hline
\newdaly & 26.8 & 4.88 & 26.8  & 4.88\\
\hline
\OptimalPrediction & 24.4 (9\%) & 3.89 (20\%)&   25.2 (6\%) &  4.44 (9\%)\\
\inexact & 24.7 (8\%) & 4.20 (14\%)&  25.5 (5\%) &  4.73 (3\%) \\
\end{tabular}
\caption{Job execution times with failures based on the failure log of LANL18 cluster, and gains due to the fault
  predictor (with respect to the performance of \newdaly).}
\label{makespan.trace18.tab}
\end{table}

\begin{table}[ht]
\centering
\begin{tabular}{c||c|c||c|c||}
 & \multicolumn{2}{|c||}{Execution time (in days)} & \multicolumn{2}{|c||}{Execution time (in days)}\\
\Cp = \Cr  & \multicolumn{2}{|c||}{($\precision = 0.82$,
  $\recall=0.85$)} & \multicolumn{2}{|c||}{($\precision = 0.4$, $\recall=0.7$)}\\
 & $2^{14}$ procs & $2^{17}$ procs & $2^{14}$ procs & $2^{17}$ procs \\
\hline
\newdaly & 26.8 & 4.86 & 26.8   & 4.86\\
\hline
\OptimalPrediction & 24.4 (9\%) & 3.85 (21\%)& 25.2 (6\%) & 4.42 (9\%)\\
\inexact & 24.6 (8\%) & 4.14 (15\%)&  25.4 (5\%) & 4.71 (3\%) \\
\end{tabular}
\caption{Job execution times with failures based on the failure log of LANL19 cluster, and gains due to the fault
  predictor (with respect to the performance of \newdaly).}
\label{makespan.trace19.tab}
\end{table}

\subsection{Recall vs. precision}
\label{section.impact}

In this section, we assess the impact of the two key parameters of the predictor, its recall \recall and its precision $\precision$.
To this purpose, we conduct simulations with synthetic traces,
where one parameter is fixed while the other varies. We choose two platforms,
a smaller one with  $N=2^{16}$ processors (or a MTBF $\mu=1,000\ min$)
and a larger one with
$N=2^{19}$ processors (or a MTBF $\mu=125\ min$).  In both cases we
study the impact of the predictor characteristics assuming a Weibull 
fault distribution with shape parameter either $0.5$ or $0.7$, under the
scenario $\Cp=\Cr$.

In Figures~\ref{fig.recall.07} and~\ref{fig.recall.05}, we fix the value of \recall (either 
$\recall=0.4$ or $\recall=0.8$) and we let $\precision$ vary from $0.3$ to $0.99$.  
In the four plots, we observe that the precision has a minor impact on
the waste, whether it is with a Weibull distribution of shape parameter $0.7$ (Figure~\ref{fig.recall.07}), 
or a Weibull distribution of shape parameter $0.5$ (Figure~\ref{fig.recall.05}).
In Figures~\ref{fig.precision.07} and~\ref{fig.precision.05}, 
we conduct the converse experiment and  fix the value of \precision (either $ \precision=0.4$
or  $\precision=0.8$), letting $\recall$ vary from $0.3$ to $0.99$.  Here we observe that 
increasing the recall significantly improves performance, in all but
one configuration. In the configuration where improving the recall does
not make a (significant) difference, there is a very large number of faults and a low
precision, hence a large number of false predictions which negatively
impact the performance whatever the value of the recall.

Altogether we conclude that it is more important (for the design of future predictors) to focus on 
improving the recall \recall rather than the precision \precision, and our results can help quantify 
this statement. We provide an intuitive explanation as follows: unpredicted faults prove very 
harmful and heavily increase the waste, while unduly checkpointing due
to false predictions (usually) turns out 
to induce a smaller overhead. 
 

 \begin{figure*}
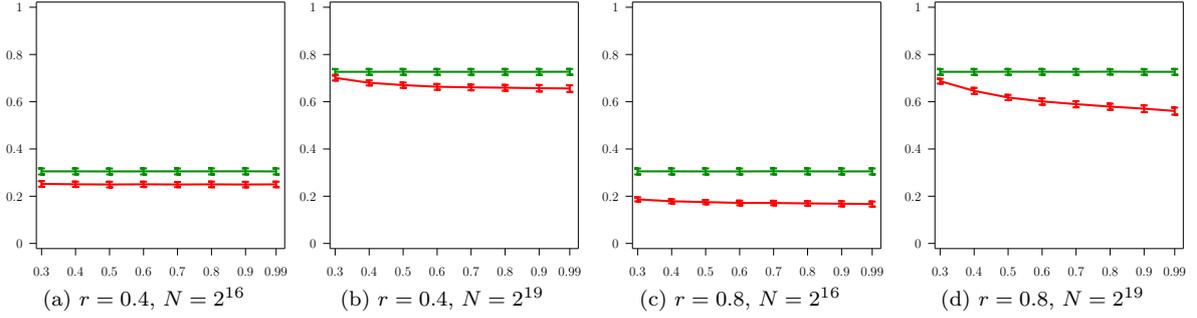

\centering
\resizebox{\textwidth}{!}{
\subfloat[$\recall=0.4$, $N=2^{16}$]
{
\includegraphics[scale=0.41]{r04platform16-WEIBULL-07-fixedC-appli0-platform-variation.tex}
}
\subfloat[$\recall=0.4$, $N=2^{19}$]
{
\includegraphics[scale=0.41]{r04platform19-WEIBULL-07-fixedC-appli0-platform-variation.tex}
}
\subfloat[$\recall=0.8$, $N=2^{16}$]
{
\includegraphics[scale=0.41]{r08platform16-WEIBULL-07-fixedC-appli0-platform-variation.tex}
}
\subfloat[$\recall=0.8$, $N=2^{19}$]
{
\includegraphics[scale=0.41]{r08platform19-WEIBULL-07-fixedC-appli0-platform-variation.tex}
}
}
\caption{Waste (y-axis) as a function of the precision (x-axis) for a fixed recall ($\recall=0.4$ and
  $\recall=0.8$) and for a Weibull distribution of faults (with shape
  parameter $k=0.7$).}
	\label{fig.recall.07}
\end{figure*}
 
 \begin{figure*}
 \centering
 \resizebox{\textwidth}{!}{
 \subfloat[$\recall=0.4$, $N=2^{16}$]
 {
 \includegraphics[scale=0.41]{r04platform16-WEIBULL-05-fixedC-appli0-platform-variation.tex}
 }
 \subfloat[$\recall=0.4$, $N=2^{19}$]
 {
 \includegraphics[scale=0.41]{r04platform19-WEIBULL-05-fixedC-appli0-platform-variation.tex}
 }
 \subfloat[$\recall=0.8$, $N=2^{16}$]
 {
 \includegraphics[scale=0.41]{r08platform16-WEIBULL-05-fixedC-appli0-platform-variation.tex}
 }
 \subfloat[$\recall=0.8$, $N=2^{19}$]
 {
 \includegraphics[scale=0.41]{r08platform19-WEIBULL-05-fixedC-appli0-platform-variation.tex}
 }
 }
 \caption{Waste (y-axis) as a function of the precision (x-axis) for a fixed recall ($\recall=0.4$
   and  $\recall=0.8$) and for a Weibull distribution of faults (with
   shape parameter $k=0.5$).}
 	\label{fig.recall.05}
 \end{figure*}

\begin{figure*}
\centering
\resizebox{\textwidth}{!}{
\subfloat[$\precision=0.4$, $N=2^{16}$]
{
\includegraphics[scale=0.41]{p04platform16-WEIBULL-07-fixedC-appli0-platform-variation.tex}
}
\subfloat[$\precision=0.4$, $N=2^{19}$]
{
\includegraphics[scale=0.41]{p04platform19-WEIBULL-07-fixedC-appli0-platform-variation.tex}
}
\subfloat[$\precision=0.8$, $N = 2^{16}$]
{
\includegraphics[scale=0.41]{p08platform16-WEIBULL-07-fixedC-appli0-platform-variation.tex}
}
\subfloat[$\precision=0.8$, $N=2^{19}$]
{
\includegraphics[scale=0.41]{p08platform19-WEIBULL-07-fixedC-appli0-platform-variation.tex}
}
}
\caption{Waste (y-axis) as a function of the recall (x-axis) for a fixed precision ($\precision=0.4$ and  $\precision=0.8$) and for a Weibull distribution (k=0.7).}
	\label{fig.precision.07}
\end{figure*}

 \begin{figure*}
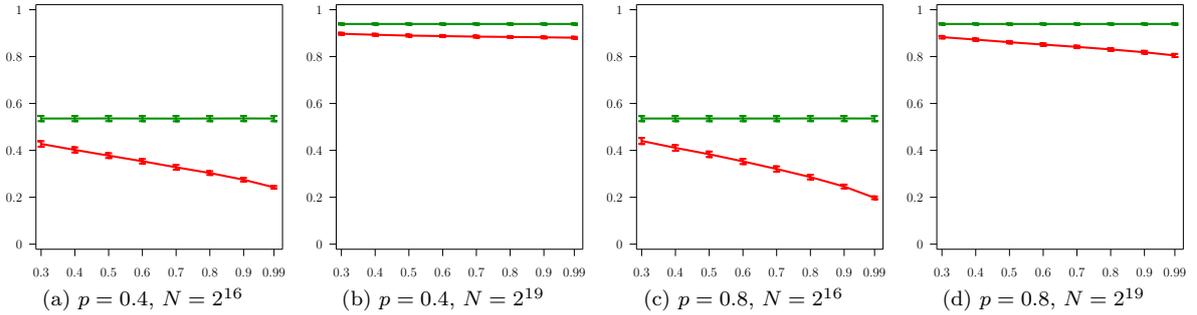

 \centering
 \resizebox{\textwidth}{!}{
 \subfloat[$\precision=0.4$, $N=2^{16}$]
 {
 \includegraphics[scale=0.41]{p04platform16-WEIBULL-05-fixedC-appli0-platform-variation.tex}
 }
 \subfloat[$\precision=0.4$, $N=2^{19}$]
 {
 \includegraphics[scale=0.41]{p04platform19-WEIBULL-05-fixedC-appli0-platform-variation.tex}
 }
 \subfloat[$\precision=0.8$, $N = 2^{16}$]
 {
 \includegraphics[scale=0.41]{p08platform16-WEIBULL-05-fixedC-appli0-platform-variation.tex}
 }
 \subfloat[$\precision=0.8$, $N=2^{19}$]
 {
 \includegraphics[scale=0.41]{p08platform19-WEIBULL-05-fixedC-appli0-platform-variation.tex}
 }
 }
 \caption{Waste (y-axis) as a function of the recall (x-axis) for a fixed precision ($\precision=0.4$ and  $\precision=0.8$) and for a Weibull distribution (k=0.5). }
 	\label{fig.precision.05}
 \end{figure*}

\section{Related work}
\label{sec.related}

Considerable research has been devoted to fault prediction, using very
different models (system log analysis~\cite{5958823}, event-driven
approach~\cite{GainaruIPDPS12,5958823,5542627}, support vector
machines~\cite{LiangZXS07,Fulp:2008:PCS:1855886.1855891}, nearest
neighbors~\cite{LiangZXS07}, etc).  In this section we give a brief
overview of existing predictors, focusing on their characteristics
rather than on the methods of prediction. For the sake of clarity, we
sum up the characteristics of the different fault predictors
from the literature in Table~\ref{rel.work.tab}.

The authors of~\cite{5542627} introduce the \emph{lead time}, that is
the duration between the time the prediction is made and the time the
predicted fault is supposed to happen. This time should be
sufficiently large to enable proactive actions. As already mentioned,
the distribution of lead times is irrelevant. Indeed, only predictions
whose lead time is greater than \Cp, the time to take a proactive
checkpoint, are meaningful. Predictions whose lead time is smaller
than \Cp, whenever they materialize as actual faults, should be
classified as unpredicted faults; the predictor recall should be
decreased accordingly.

The predictor of~\cite{5542627} is also able to locate where the
predicted fault is supposed to strike. This additional characteristics
has a negative impact on the precision (because a fault happening at
the predicted time but not on the predicted location is classified as
a non predicted fault; see the low value of \precision in
Table~\ref{rel.work.tab}). The authors of~\cite{5542627} state that
fault localization has a positive impact on proactive checkpointing
time in their context: instead of a full checkpoint costing $1,500$
seconds they can take a partial checkpoint costing only $12$
seconds. This led us to introduce a different cost \Cp for proactive
checkpoints, that can be smaller than the cost \Cr of regular
checkpoints. Gainaru et al.~\cite{GainaruSC12} also stated that
fault-localization could help decrease the checkpointing time. Their
predictor also gives information on fault localization.  They studied
the impact of different lead times on the recall of their predictor.
Papers~\cite{5958823} and~\cite{LiangZXS07} also considered lead
times.

\begin{table}
\centering
\begin{tabular}{|c|c|c|c|}
\hline
Paper & Lead Time & Precision & {Recall}  \\
\hline
\cite{5542627} & 300 s & 40 \% & 70 \%  \\
\cite{5542627} & 600 s & 35 \% & 60 \%  \\
\cite{5958823} & 2h & 64.8 \% & 65.2 \% \\
\cite{5958823} & 0 min & 82.3 \% & 85.4  \% \\
\cite{GainaruIPDPS12} & 32 s & 93 \% & 43  \% \\
\cite{GainaruSC12} & 10s & 92 \% & 40 \% \\
\cite{GainaruSC12} & 60s & 92 \% & 20 \% \\
\cite{GainaruSC12} & 600s & 92 \% & 3 \% \\
\cite{Fulp:2008:PCS:1855886.1855891} & NA & 70 \% & 75  \% \\
\cite{LiangZXS07} & NA & 20 \% & 30  \% \\
\cite{LiangZXS07} & NA & 30 \%& 75  \% \\
\cite{LiangZXS07} & NA & 40 \%& 90  \% \\
\cite{LiangZXS07} & NA & 50 \% & 30  \% \\
\cite{LiangZXS07} & NA & 60 \% & 85 \% \\
\hline
\end{tabular}
\caption{Comparative study of different parameters returned by some predictors.}
\label{rel.work.tab}
\end{table}

Most studies on fault prediction state that a proactive action must be
taken right before the predicted fault, be it a checkpoint
or a migration.  However, we have shown in this paper that it is
beneficial to ignore some predictions, namely when the predicted fault is
announced to strike less than $\frac{\Cp}{\precision}$ seconds after the
last periodic checkpoint. 

Gainaru et al.~\cite{GainaruSC12} studied the impact of
prediction on the checkpointing period. Their computation of the total waste
is not fully accurate and they do not provide any minimization analysis.  
Instead, they only propose to use Young's formula, replacing the MTBF by the 
mean-time of unpredicted faults. They do not question whether all predictions should be taken into account.
Furthermore, they did not conduct any simulations;
instead they analytically computed the ratio of the waste with and without predictions
and instantiated the corresponding formula with several scenarios.

Li et al.~\cite{li2009fault} considered the mathematical problem of
when and how to migrate. In order to be able to use migration, they
assumed that at any time 2\% of the resources are available as
spares. This allows them to conceive a Knapsack-based
heuristic. Thanks to their algorithm, they were able to save 30\% of
the execution time compared to a heuristic that does not take the
prediction into account, with a precision and recall of 70\%, and with
a maximum load of 0.7. In our study we do not consider that we have
a batch of spare resources. We assume that after a downtime the
resources that failed are once again available.

Note that some authors \cite{5958823,LiangZXS07} do not consider that
their predictors predict the exact time of the fault. On the contrary,
they consider a ``prediction window'' which is the time interval in
which the predicted is supposed to occur. Because most papers focus on
prediction windows of negligible length, we did not consider
prediction windows in this study.

Finally, to the best of our knowledge, this work is the first to focus on the mathematical aspect of 
fault prediction, and to provide a model and a detailed analysis of the waste due to all three types of
events (true and false predictions and unpredicted failures). 

\section{Conclusion}
\label{sec.conclusion}

In this work we have studied the impact of fault prediction on
periodic checkpointing. We started by revisiting the first-order
approach by Young and Daly. We have performed a refined analysis
leading to a better checkpointing period:  \Tyafo  is slightly closer to the optimal period
for Exponential distributions (the only case where the optimal is known),
and leads to smaller execution times for Weibull distributions (as shown in Section~\ref{sec.simulations.synthetic}).
 
Then we have extended the analysis to include
fault predictions. We have established analytical conditions stating whether a fault
prediction should be taken into account  or not. More importantly, we 
have proven that the optimal approach is to never trust the predictor in the beginning of a regular 
period, and to always trust it in the end of the period; the cross-over point $\frac{\Cp}{\precision}$
depends on the time to take a proactive checkpoint and on the precision of the predictor.
This striking result is somewhat unexpected, as one might have envisioned more trust regimes,
with several intermediate trust levels smoothly evolving from a ``never trust'' policy to an ``always trust'' one.

We have conducted simulations involving synthetic failure traces following either an Exponential distribution law or a Weibull one. We have also
used log-based failure traces. In addition, we have used exact prediction dates and uncertainty intervals for these dates. Through this extensive experiment setting,
we have established the accuracy of the model, of its analysis, and of the
predicted period (in the presence of a fault predictor). The
simulations also show that even a not-so-good  fault predictor
can lead to quite a  significant decrease in the  application execution
time.  We have also shown that the most important
characteristic of a fault predictor is its recall (the
percentage of actually predicted faults) rather than its precision
(the percentage of predictions that actually correspond to faults):
\emph{better safe than sorry}, or better prepare for a false event than miss an actual failure!

Altogether, the analytical model and the comprehensive results provided in this work enable to
fully assess the impact of fault prediction on optimal checkpointing strategies. 
Future work will be devoted to the study of the impact of fault prediction on
uncoordinated or hierarchical checkpointing protocols. Another challenging problem
is to determine the best trade-off between performance and energy consumption
when combining several resilience techniques 
such as checkpointing, prediction, and replication.

\section*{Acknowledgments.} 
Y.~Robert is with the Institut Universitaire de France.
This work was supported in part by the ANR {\em RESCUE} project.
We would like to thank the reviewers of JPDC for their
comments and suggestions, which greatly helped improve the final version of the paper.
\bigskip
\bibliographystyle{elsarticle-num}
\bibliography{biblio}

\begin{thebibliography}{10}
\expandafter\ifx\csname url\endcsname\relax
  \def\url#1{\texttt{#1}}\fi
\expandafter\ifx\csname urlprefix\endcsname\relax\def\urlprefix{URL }\fi
\expandafter\ifx\csname href\endcsname\relax
  \def\href#1#2{#2} \def\path#1{#1}\fi

\bibitem{6264677}
G.~Zheng, X.~Ni, L.~Kale, A scalable double in-memory checkpoint and restart
  scheme towards exascale, in: Dependable Systems and Networks Workshops
  (DSN-W), 2012.
\newblock \href {http://dx.doi.org/10.1109/DSNW.2012.6264677}
  {\path{doi:10.1109/DSNW.2012.6264677}}.

\bibitem{TsubameSC12}
K.~Sato, A.~Moody, K.~Mohror, T.~Gamblin, B.~R. de~Supinski, N.~Maruyama,
  S.~Matsuoka, Design and modeling of a non-blocking checkpointing system, in:
  SC'12 (the 2012 International Conference for High Performance Computing,
  Networking, Storage and Analysis), 2012.

\bibitem{Fulp:2008:PCS:1855886.1855891}
E.~W. Fulp, G.~A. Fink, J.~N. Haack, Predicting computer system failures using
  support vector machines, in: Proceedings of the First USENIX conference on
  Analysis of system logs, USENIX Association, 2008.

\bibitem{GainaruIPDPS12}
A.~Gainaru, F.~Cappello, W.~Kramer, Taming of the shrew: Modeling the normal
  and faulty behavior of large-scale hpc systems, in: Proc. IPDPS'12, 2012.

\bibitem{GainaruSC12}
A.~Gainaru, F.~Cappello, W.~Kramer, M.~Snir, Fault prediction under the
  microscope - a closer look into hpc systems, in: SC'12 (the 2012
  International Conference for High Performance Computing, Networking, Storage
  and Analysis), 2012.

\bibitem{LiangZXS07}
Y.~Liang, Y.~Zhang, H.~Xiong, R.~K. Sahoo, Failure prediction in ibm bluegene/l
  event logs, in: ICDM, 2007, pp. 583--588.

\bibitem{5958823}
L.~Yu, Z.~Zheng, Z.~Lan, S.~Coghlan, Practical online failure prediction for
  blue gene/p: Period-based vs event-driven, in: Dependable Systems and
  Networks Workshops (DSN-W), 2011, pp. 259--264.
\newblock \href {http://dx.doi.org/10.1109/DSNW.2011.5958823}
  {\path{doi:10.1109/DSNW.2011.5958823}}.

\bibitem{5542627}
Z.~Zheng, Z.~Lan, R.~Gupta, S.~Coghlan, P.~Beckman, A practical failure
  prediction with location and lead time for blue gene/p, in: Dependable
  Systems and Networks Workshops (DSN-W), 2010, pp. 15--22.
\newblock \href {http://dx.doi.org/10.1109/DSNW.2010.5542627}
  {\path{doi:10.1109/DSNW.2010.5542627}}.

\bibitem{young74}
J.~W. Young, {A first order approximation to the optimum checkpoint interval},
  Comm. of the ACM 17~(9) (1974) 530--531.

\bibitem{daly04}
J.~T. Daly, A higher order estimate of the optimum checkpoint interval for
  restart dumps, FGCS 22~(3) (2004) 303--312.

\bibitem{875631}
N.~Kolettis, N.~D. Fulton, Software rejuvenation: Analysis, module and
  applications, in: FTCS '95, IEEE CS, Washington, DC, USA, 1995, p. 381.

\bibitem{1663301}
V.~Castelli, R.~E. Harper, P.~Heidelberger, S.~W. Hunter, K.~S. Trivedi,
  K.~Vaidyanathan, W.~P. Zeggert, Proactive management of software aging, IBM
  J. Res. Dev. 45~(2) (2001) 311--332.

\bibitem{Hong01}
J.~Hong, S.~Kim, Y.~Cho, H.~Yeom, T.~Park, On the choice of checkpoint interval
  using memory usage profile and adaptive time series analysis, in: {Proc.
  Pacific Rim Int. Symp. on Dependable Computing}, IEEE Computer Society, 2001.

\bibitem{wingstrom-phd}
J.~Wingstrom, {Overcoming The Difficulties Created By The Volatile Nature Of
  Desktop Grids Through Understanding, Prediction And Redundancy}, Ph.D.
  thesis, University of Hawai`i at Manoa (2009).

\bibitem{SC2011}
M.~Bougeret, H.~Casanova, M.~Rabie, Y.~Robert, F.~Vivien, Checkpointing
  strategies for parallel jobs, in: Proceedings of SC'11, 2011.

\bibitem{c183}
Y.~Robert, F.~Vivien, D.~Zaidouni, On the complexity of scheduling checkpoints
  for computational workflowss, in: {FTXS'2012}, the Workshop on
  Fault-Tolerance for HPC at Extreme Scale, in conjunction with the 42nd Annual
  IEEE/IFIP Int. Conf. on Dependable Systems and Networks (DSN 2012), {IEEE}
  Computer Society Press, 2012.

\bibitem{Mitzenmacher2005}
M.~Mitzenmacher, E.~Upfal, Probability and Computing: Randomized Algorithms and
  Probabilistic Analysis, Cambridge University Press, 2005.

\bibitem{Weibull1}
T.~Heath, R.~P. Martin, T.~D. Nguyen, Improving cluster availability using
  workstation validation, SIGMETRICS Perf. Eval. Rev. 30~(1).

\bibitem{Weibull2}
B.~Schroeder, G.~A. Gibson, A large-scale study of failures in high-performance
  computing systems, in: Proc. of DSN, 2006, pp. 249--258.

\bibitem{liu2008optimal}
Y.~Liu, R.~Nassar, C.~Leangsuksun, N.~Naksinehaboon, M.~Paun, S.~Scott, {An
  optimal checkpoint/restart model for a large scale high performance computing
  system}, in: IPDPS'08, IEEE, 2008.

\bibitem{Heien:2011:MTH:2063384.2063444}
E.~Heien, D.~Kondo, A.~Gainaru, D.~LaPine, B.~Kramer, F.~Cappello, Modeling and
  tolerating heterogeneous failures in large parallel systems, in: Proc.
  ACM/IEEE Supercomputing'11, ACM Press, 2011.

\bibitem{10.1109/CCGRID.2010.71}
D.~Kondo, B.~Javadi, A.~Iosup, D.~Epema, The failure trace archive: Enabling
  comparative analysis of failures in diverse distributed systems, Cluster
  Computing and the Grid, IEEE International Symposium on 0 (2010) 398--407.
\newblock \href
  {http://dx.doi.org/http://doi.ieeecomputersociety.org/10.1109/CCGRID.2010.71}
  {\path{doi:http://doi.ieeecomputersociety.org/10.1109/CCGRID.2010.71}}.

\bibitem{j116}
F.~Cappello, H.~Casanova, Y.~Robert, Preventive migration vs. preventive
  checkpointing for extreme scale supercomputers, Parallel Processing Letters
  21~(2) (2011) 111--132.

\bibitem{Ferreira2011}
K.~Ferreira, J.~Stearley, J.~H.~I. Laros, R.~Oldfield, K.~Pedretti,
  R.~Brightwell, R.~Riesen, P.~G. Bridges, D.~Arnold, {Evaluating the Viability
  of Process Replication Reliability for Exascale Systems}, in: Proceedings of
  the 2011 ACM/IEEE Conf. on Supercomputing, 2011.

\bibitem{li2009fault}
Y.~Li, Z.~Lan, P.~Gujrati, X.~Sun, Fault-aware runtime strategies for
  high-performance computing, Parallel and Distributed Systems, IEEE
  Transactions on 20~(4) (2009) 460--473.

\bibitem{Ross10}
S.~M. Ross, Introduction to Probability Models, Tenth Edition, Academic Press,
  2009.

\end{thebibliography}

\appendix
\section{}
	\label{app.mu}

For the sake of completeness, we provide a proof of the following result:

\begin{proposition}
Consider a platform comprising $N$ components, and 
assume that the inter-arrival times of the faults on the components 
are independent and identically distributed random variables that follow an arbitrary probability
law whose expectation is $\mu_{\text{ind}}$. 
Then the expectation of the inter-arrival times of the faults on the whole platform is
$\mu =
\frac{\mu_{\text{ind}}}{N}$. 
\end{proposition}

\begin{proof}
\vfill
\noindent
Consider first a single component, say component number $q$. Let
  $X_{i}$, $i \geq 0$ denote the IID random variables for fault inter-arrival times
on that component, with $\esperance{X_{i}} = \mu_{\text{ind}}$.
Consider a fixed time bound $F$. 
Let $n_{q}(F)$  be the number of faults on the component until time $F$ is
exceeded. In other words, the $(n_{q}(F)-1)$-th fault is the last one to happen
strictly before time $F$, and the $n_{q}(F)$-th fault is the first to happen
at time $F$ or after.
By definition of $n_{q}(F)$, we have 
$$\sum_{i=1}^{n_{q}(F)-1} X_{i} \leq F \leq \sum_{i=1}^{n_{q}(F)} X_{i}$$
Using Wald's equation~\cite[p. 486]{Ross10}, with $n_{q}(F)$ as a stopping criterion,
we derive:
$$(\esperance{n_{q}(F)}-1) \mu_{\text{ind}} \leq F \leq \esperance{n_{q}(F)} \mu_{\text{ind}}$$
and we obtain:
\begin{equation}
\lim_{F \rightarrow +\infty} \frac{\esperance{n_{q}(F)}}{F} = \frac{1}{\mu_{\text{ind}}}
\label{eq.wald}
\end{equation}

Consider now the whole platform, and  
let  $Y_{i}$, $i \geq 0$ denote the IID random variables for fault inter-arrival times
on the platform, with $\esperance{Y_{i}} = \mu$.
Consider 
a fixed time bound $F$ as before.
Let $n(F)$  be the number of faults on the whole platform until time $F$ is exceeded.
With the same reasoning for the whole platform as for a single
component, we derive:
\begin{equation}
\lim_{F \rightarrow +\infty} \frac{\esperance{n(F)}}{F} = \frac{1}{\mu}
\label{eq.wald1}
\end{equation}
Now let 
$m_{q}(F)$ be the number of these faults that strike component number $q$.
Of course we have $n(F) = \sum_{q=1}^{N} m_{q}(F)$.
By definition, except for the component hit by the last failure, 
$m_{q}(F)+1 $  is the number of failures on component $q$ until time $F$ is exceeded,
hence $n_{q}(F) = m_{q}(F)+1$
(and this number is $m_{q}(F) = n_{q}(F)$ on the component hit by the last failure).
From Equation~\eqref{eq.wald} again,  we have for each component $q$:
$$\lim_{F \rightarrow +\infty} \frac{\esperance{m_{q}(F)}}{F} = \frac{1}{\mu_{\text{ind}}}$$
Since $n(F) = \sum_{q=1}^{N} m_{q}(F)$, we also have:
\begin{equation}
\lim_{F \rightarrow +\infty} \frac{\esperance{n(F)}}{F} = \frac{N}{\mu_{\text{ind}}}
\label{eq.wald2}
\end{equation}
Equations~\eqref{eq.wald1} and ~\eqref{eq.wald2} lead to the result.
\end{proof}

\section{}
	\label{app.sup.plots}

        In this section, we provide results for synthetic failure
        traces when false predictions are generated according to a
        uniform distribution.

\begin{figure*}
\centering
\includegraphics[scale=0.5]{newlegend.tex}\\
\hspace{-1cm}
\subfloat[Maple]
{
\includegraphics[scale=0.16]{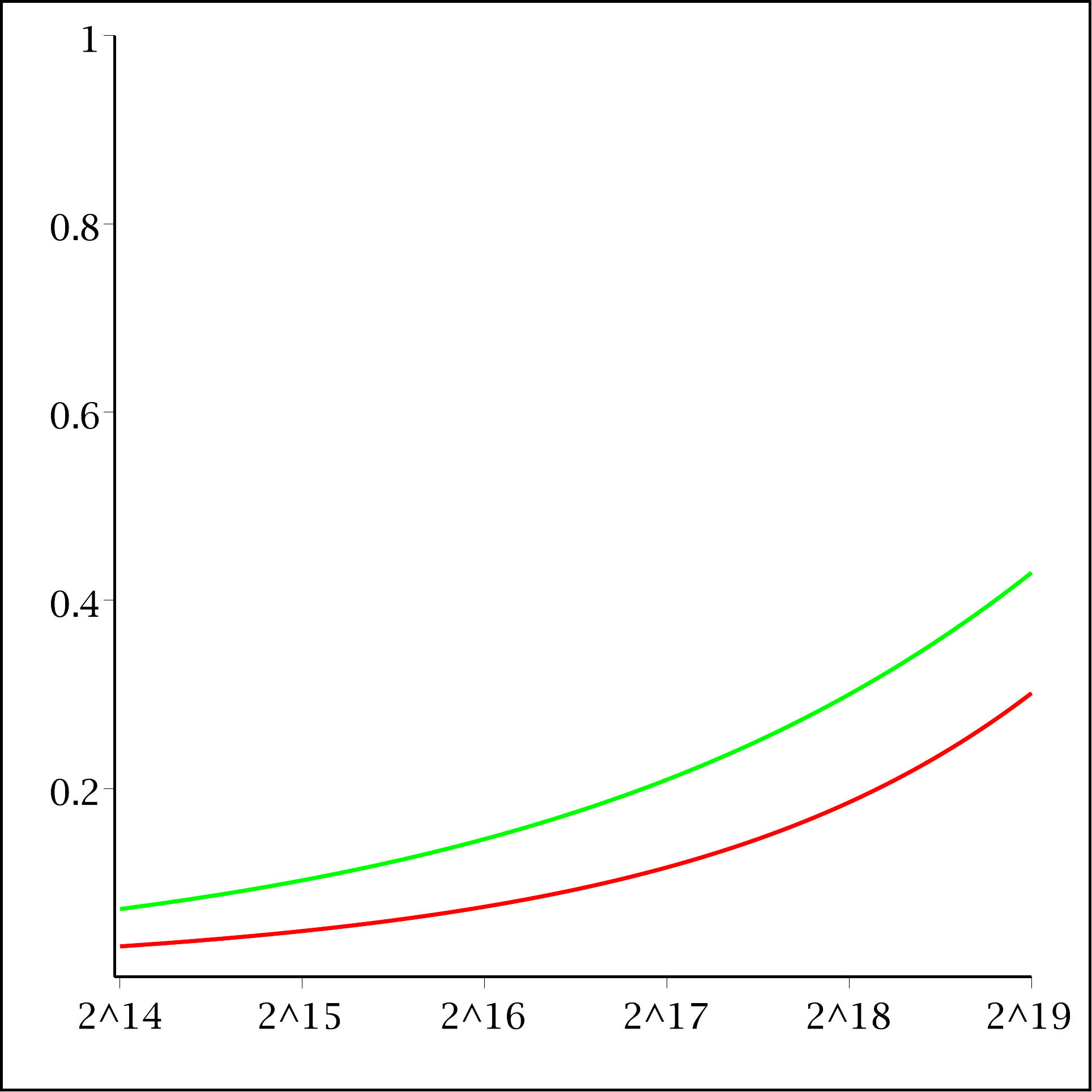}
}
\subfloat[Exponential]
{
\includegraphics[scale=0.36]{r085p083UNIF-EXP-fixedC-appli0-platform-variation.tex}
}	
\subfloat[Weibull $k=0.7$]
{
\includegraphics[scale=0.36]{r085p083UNIF-WEIBULL-07-fixedC-appli0-platform-variation.tex}
}
\subfloat[Weibull $k=0.5$]
{
\includegraphics[scale=0.36]{r085p083UNIF-WEIBULL-05-fixedC-appli0-platform-variation.tex}
}
\\
\hspace{-1cm}
\subfloat[Maple]
{
\includegraphics[scale=0.16]{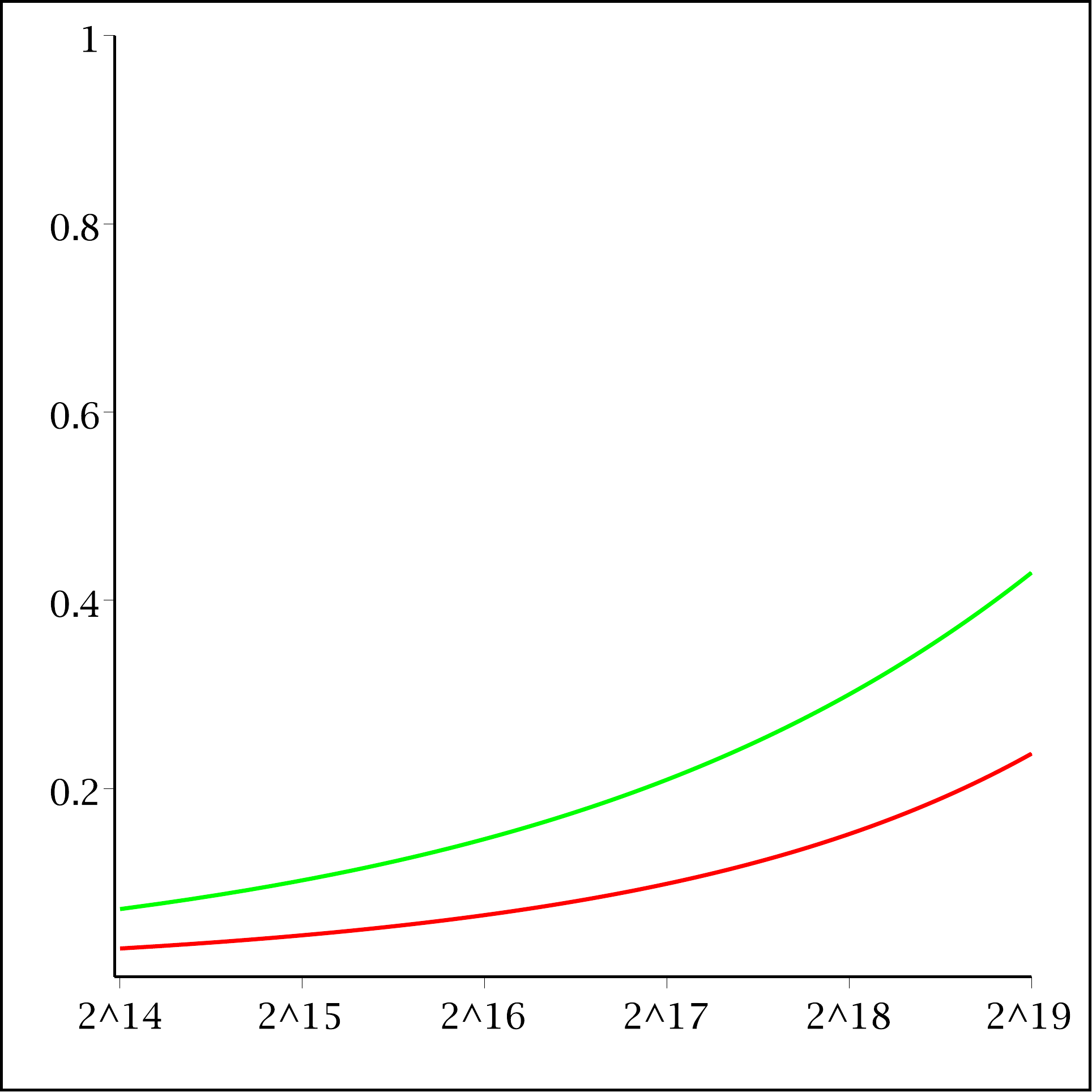}
}
\subfloat[Exponential]
{
\includegraphics[scale=0.36]{r085p083alpha1sur10UNIF-EXP-fixedC-appli0-platform-variation.tex}
}	
\subfloat[Weibull $k=0.7$]
{
\includegraphics[scale=0.36]{r085p083aplha1sur10UNIF-WEIBULL-07-fixedC-appli0-platform-variation.tex}
}
\subfloat[Weibull $k=0.5$]
{
\includegraphics[scale=0.36]{r085p083aplha1sur10UNIF-WEIBULL-05-fixedC-appli0-platform-variation.tex}
}
\\
\hspace{-1cm}
\subfloat[Maple]
{
\includegraphics[scale=0.16]{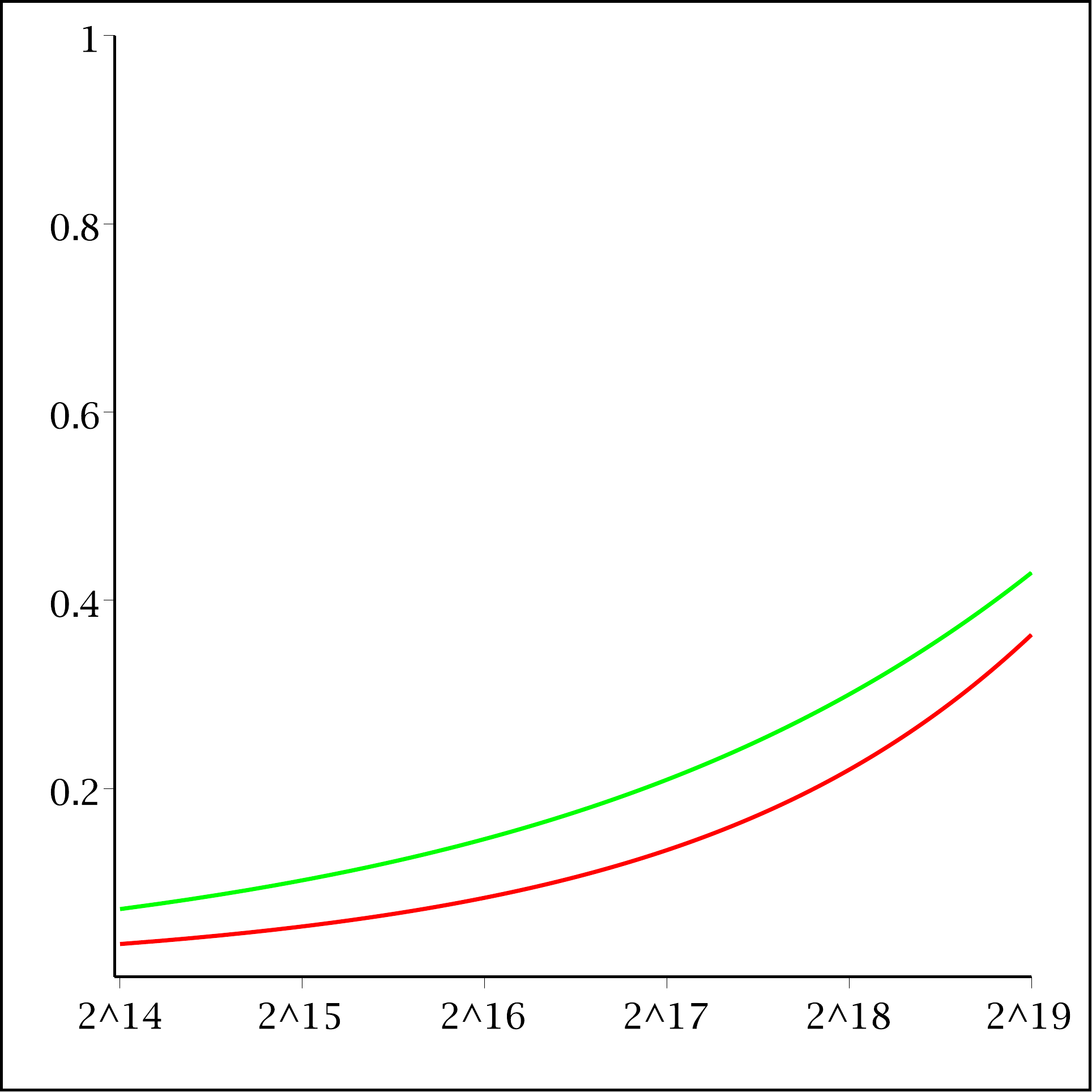}
}
\subfloat[Exponential]
{
\includegraphics[scale=0.36]{r085p083alpha2UNIF-EXP-fixedC-appli0-platform-variation.tex}
}	
\subfloat[Weibull $k=0.7$]
{
\includegraphics[scale=0.36]{r085p083alpha2UNIF-WEIBULL-07-fixedC-appli0-platform-variation.tex}
}
\subfloat[Weibull $k=0.5$]
{
\includegraphics[scale=0.36]{r085p083aplha2UNIF-WEIBULL-05-fixedC-appli0-platform-variation.tex}
}
\caption{Waste (y-axis) for the different heuristics as a function of the
  platform size (x-axis), with $\precision=0.82$, $\recall=0.85$, $\Cp=\Cr$ (first row), $\Cp=0.1 \Cr$  (second row), or $\Cp=2 \Cr$  (third  row) and with a trace of false predictions parametrized by a uniform distribution.}
	\label{fig.082.085.appendix}
\end{figure*}

\begin{figure*}
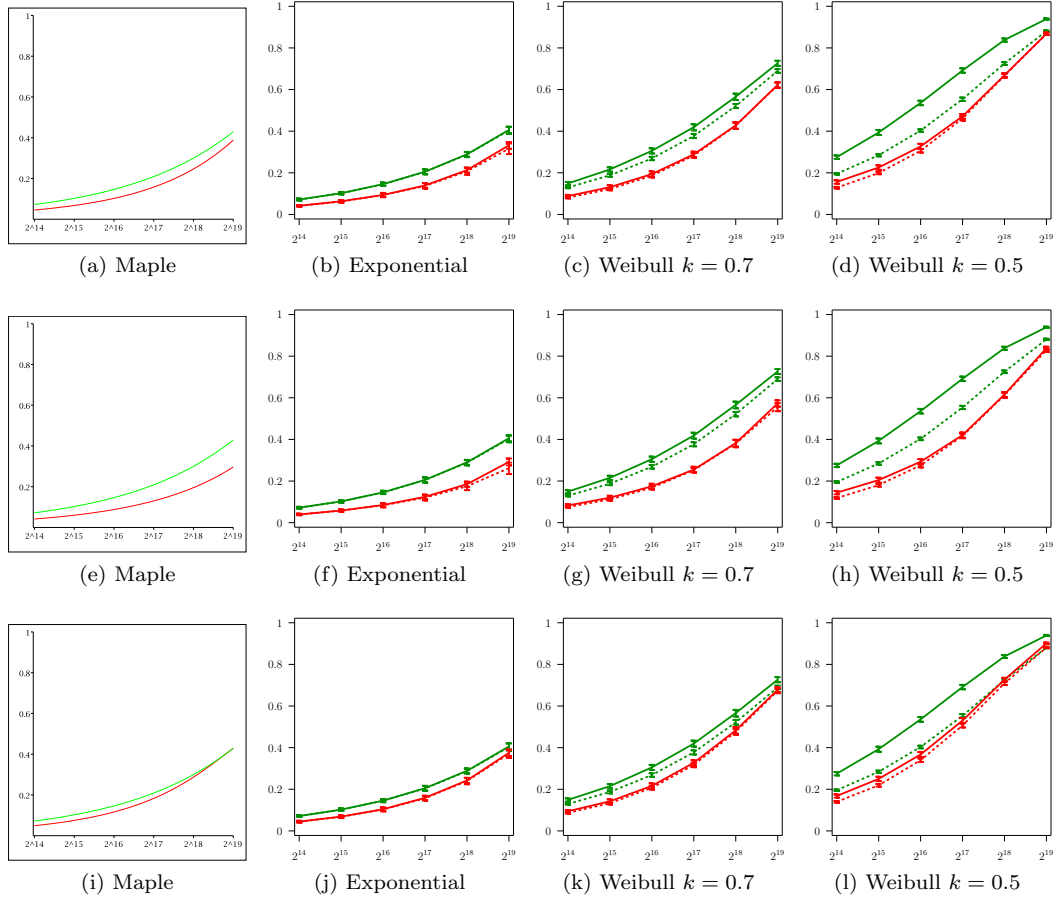

\centering
\hspace{-1cm}
\subfloat[Maple]
{
\includegraphics[scale=0.16]{maple_journal_p04_r07_1.pdf}
}
\subfloat[Exponential]
{
\includegraphics[scale=0.36]{r07p04UNIF-EXP-fixedC-appli0-platform-variation.tex}
}
\subfloat[Weibull $k=0.7$]
{
\includegraphics[scale=0.36]{r07p04UNIF-WEIBULL-07-fixedC-appli0-platform-variation.tex}
}
\subfloat[Weibull $k=0.5$]
{
\includegraphics[scale=0.36]{r07p04UNIF-WEIBULL-05-fixedC-appli0-platform-variation.tex}
}
\\
\hspace{-1cm}
\subfloat[Maple]
{
\includegraphics[scale=0.16]{maple_journal_p04_r07_01.pdf}
}
\subfloat[Exponential]
{
\includegraphics[scale=0.36]{r07p04aplha1sur10UNIF-EXP-fixedC-appli0-platform-variation.tex}
}
\subfloat[Weibull $k=0.7$]
{
\includegraphics[scale=0.36]{r07p04alpha1sur10UNIF-WEIBULL-07-fixedC-appli0-platform-variation.tex}
}
\subfloat[Weibull $k=0.5$]
{
\includegraphics[scale=0.36]{r07p04alpha1sur10UNIF-WEIBULL-05-fixedC-appli0-platform-variation.tex}
}
\\
\hspace{-1cm}
\subfloat[Maple]
{
\includegraphics[scale=0.16]{maple_journal_p04_r07_2.pdf}
}
\subfloat[Exponential]
{
\includegraphics[scale=0.36]{r07p04aplha2UNIF-EXP-fixedC-appli0-platform-variation.tex}
}
\subfloat[Weibull $k=0.7$]
{
\includegraphics[scale=0.36]{r07p04alpha2UNIF-WEIBULL-07-fixedC-appli0-platform-variation.tex}
}
\subfloat[Weibull $k=0.5$]
{
\includegraphics[scale=0.36]{r07p04alpha2UNIF-WEIBULL-05-fixedC-appli0-platform-variation.tex}
}
\caption{Waste (y-axis) for the different heuristics as a function of the
  platform size (x-axis), with $\precision=0.4$, $\recall=0.7$,  $\Cp=\Cr$  (first row), $\Cp=0.1 \Cr$ (second row), or $\Cp=2 \Cr$ (third  row) and with a trace of false predictions parametrized by a uniform distribution..}
	\label{fig.04.07.appendix}
\end{figure*}

\end{document}